\documentclass[a4paper,11pt,onecolumn,tightenlines,superscriptaddress]{article}
\pdfoutput=1
\usepackage{preamble}
\usepackage{mathrsfs}
\usepackage[a4paper,left=3cm,right=3cm]{geometry}

\begin{document}
\title{Geometric aspects of fault tolerance}
\author{Jonathan Conrad\thanks{{\'E}cole Polytechnique F{\'e}d{\'e}rale de Lausanne, \url{jonathan.conrad@epfl.ch}}}
\date{\today}

\maketitle
\begin{abstract}
We highlight three geometric aspects of fault tolerance. First, we construct an operator-algebraic formulation of the fiber bundle framework for fault tolerance, and show that a generalization of the Knill-Laflamme conditions implies a flat connection for this framework. 
Next, we develop this framework for Gottesman-Kitaev-Preskill (GKP) codes. We show that, for each polarization type, their moduli space embeds smoothly into a Grassmannian of finite-dimensional subspaces of a regularized distribution space.
Finally, we explain how GKP codes associated with Jacobians give rise to magic states that can be obtained by code space projections of Gaussian states and identify a $\rm CCZ$-type magic state arising this way.
\end{abstract}

\begin{small}
\tableofcontents
\end{small}
\clearpage
\section*{From physical to logical quantum channels}
Quantum fault tolerance studies how logical -- macroscopic -- quantum channels can be implemented using physical -- microscopic --  quantum channels in a robust manner, and much work has been devoted to constructing examples that exhibit this behaviour~\cite{Campbell_2017,Kitaev_2003, nguyen2025quantumfaulttoleranceconstantspace, aliferis2005quantum}. 
Microscopic conditions that imply fault tolerance, i.e., robustness of logical quantum channels, have, however, been limited to specific models and error assumptions. 
In ref.~\cite{gottesman2017fibre} Gottesman and Zhang proposed the idea of lifting the property of fault tolerance to a more general, quantifiable microscopic property of evolving quantum systems by constructing a fiber bundle over the space of possible code spaces. In their framework, unitary evolution implements parallel transport, and the property of microscopic fault tolerance is associated to the existence of a flat (projective) connection on the fiber bundle. The consequence of such a definition is that logical gates are realized via monodromy, and their robustness becomes a topological property of the construction. This perspective is consistent with much recent work on the implementation of logical gates through sequences of code-switching~\cite{davydova2023quantum}. Although the Gottesman--Zhang framework provides an appealing geometric perspective, it treats unitary dynamics and was developed through specialized examples.
\begin{figure}
    \centering
    \includegraphics[width=.8\textwidth]{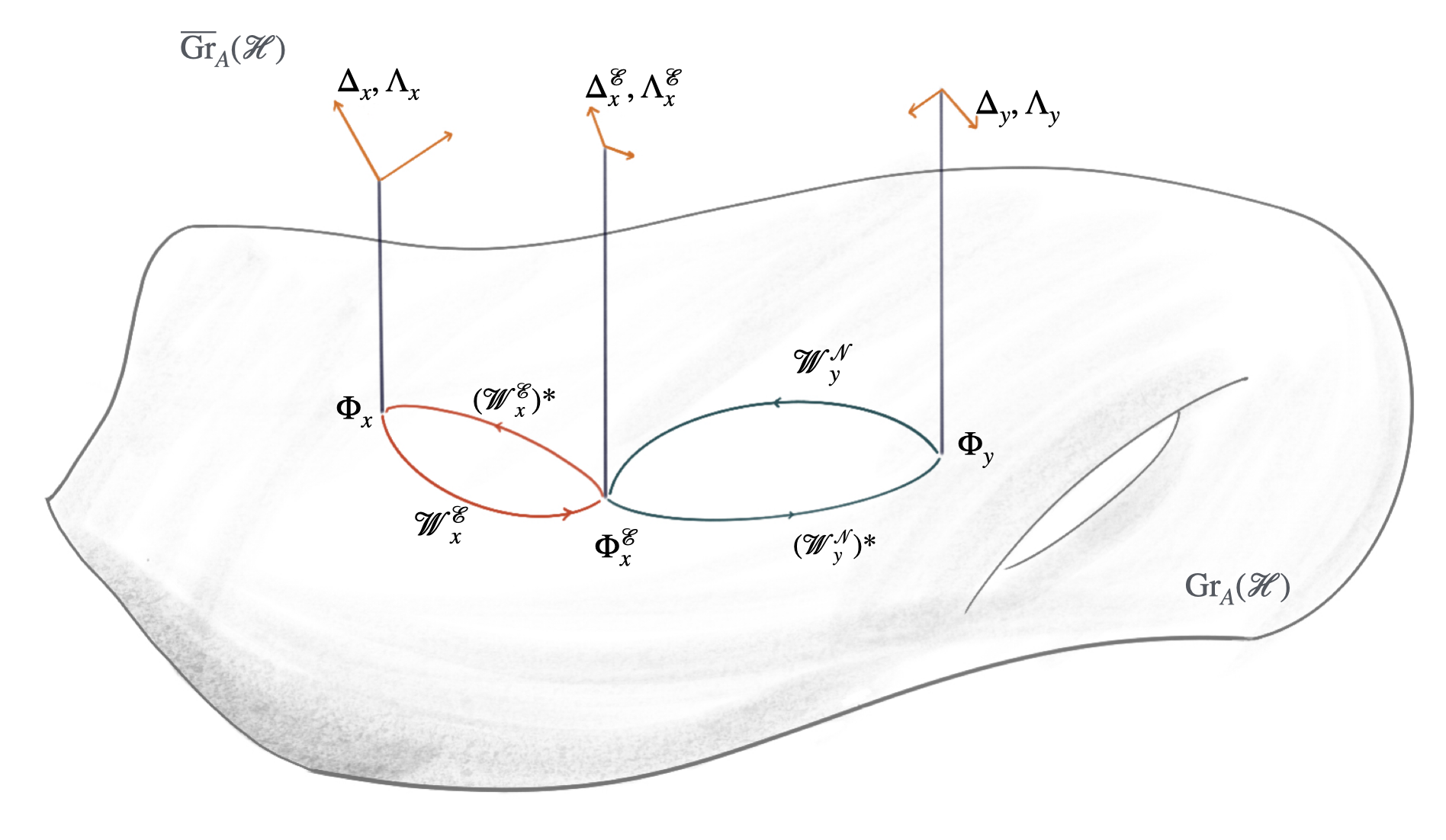}
    \caption{Illustration of the Grassmannian bundle. The total space is the framed Grassmannian $\overline{\rm Gr}_A(\CH)$ consisting of triples $(\Phi, \Delta, \Lambda)$ that project to $\Phi$ on the base space. We illustrate Theorem~\ref{thm:GeoKL}, where jointly correctable errors $\CE, \CN$ isometrically transport between the code idempotents $\Phi_x, \Phi_y$ and the error-space idempotent $\Phi^{\CE}_x$.}
\end{figure}
\subsubsection*{Aspect 1: Knill-Laflamme geometry on the Grassmannian bundle}
In this work, we extend this framework to general quantum channels. The first conceptual step is to consider the $C^*$-algebra $A$ of encoded logical operators for a physical Hilbert space $\CH$ and the Grassmannian associated to all such representations. On this Grassmannian, each encoded logical algebra is fixed by a completely positive idempotent 
\begin{equation}
    \Phi=\Phi^2=\Phi^*:\; B(\CH)\to B(\CH),
\end{equation}
and can be given a frame $(\Delta, \Lambda)$, such that $\Delta\Lambda=\Phi,\, \Lambda\Delta=\id_A$,
where the completely positive map $\Delta: A\to B(\CH)$ can be understood as the encoding isometry of the algebra $A$,\footnote{Or dually, the decoding map~\cite{kitaev2025almostidempotentquantumchannelsapproximate}.} and $\Lambda$ realizes the decoding map (we will typically consider normalized frames with $\Lambda=\Delta^*$). 
At each point, the frame realizes a representation of $A$ on $\Im \Phi$ via the product \begin{equation}X\times_{\Delta}Y=\Phi(X\Delta(1_A)^{-1}Y),\end{equation} which is similar to the Choi-Effros product examined in ref.~\cite{kitaev2025almostidempotentquantumchannelsapproximate}.
We consider transport on this bundle induced by channels satisfying a generalized set of  Knill-Laflamme (KL) conditions~\cite{knill1997theory}. For two idempotents $\Phi_x, \Phi_y$, and two quantum channels $\CE_s, \CE_t$, these are formulated as the diagonality of the Gram matrix
\begin{equation}
    (\CE_s\Phi_x)^*(\CE_t\Phi_y)=\sum_{k} \alpha_{xy, st}^k \Delta_x p_k \Lambda_y,
\end{equation}
where $\alpha_{xy, st}^k>0$ and $p_k$ are projections onto the $k$-th block of $A$.
For $x=y$ and $s=t$, this formulation is equivalent to the KL conditions and their generalizations~\cite{kribs2006operatorquantumerrorcorrection, beny2007generalization}, in which case they imply that correctable errors implement isometric transport on this Grassmannian. 
For different points $x\neq y$, this condition implies that logical information is faithfully transported. When the generalized KL conditions hold for every pair of channels in a set $S$, then it also holds for all convex combinations of channels. We can hence endow a smooth manifold structure on such sets and show that the flat connection is realized in the differential geometric sense.
\subsubsection*{Aspect 2: Chern connection on moduli space of GKP codes}
Next, we consider the class of GKP codes. These are quantum error-correcting codes for infinite-dimensional Hilbert spaces associated to symplectically integral lattices\footnote{Or polarized abelian varieties, if this perspective is more pleasing.}, and their associated Grassmannian is represented  by a space of dual-Schwartz distributions. We first develop the structure of this Grassmannian using regularized operator theory in Segal-Bargmann representation of quantum mechanics and then show that it can be identified with the moduli space of GKP-codes $\CM_D$ via a smooth embedding. These considerations serve to show that the theory developed in sec.~\ref{sec:KL_Geometry} can be applied to GKP codes and substantiates previous discussions of flat-transport of the logical algebra for GKP codes over its parameter space, as developed in refs.~\cite{conrad2024latticesgatescurvesgkp,Burchards_2025}. The bundle of code spaces over the choices of GKP codes, projected to the complex orbifold $\CJ_D=\CM_D/\!/U(n)$ is known to carry a projectively flat Chern-connection, and we explain the relationship between the different perspectives.
\subsubsection*{Aspect 3: GKP magic states from orbifold points and Jacobians}
In the final part, we investigate GKP magic states~\cite{BravyiKitaev} arising from unitary (symplectic orthogonal) symmetries. For GKP codes obtained as Jacobians of a compact Riemann surface $C$, automorphisms of $C$ induce unitary automorphisms of the GKP code. Consequently, projection of the physical vacuum onto the corresponding code spaces yields quantum states with additional symmetries, which may be magic states. We explain the general mechanism, compute the structure of the magic state associated to the Klein quartic curve, and find that a $\ket{\rm CCZ}$ magic state can be obtained from the octahedral hyperelliptic curve.

\subsection*{Related works}
This manuscript builds on a wide range of work on quantum error correction and fault tolerance. Here we summarize the central ingredients and refer to more specific literature throughout the following sections.
The ideas of the first section rely heavily on the quest put forward by Gottesman and Zhang in ref.~\cite{gottesman2017fibre} to conceive a geometric, bundle-theoretical formulation of quantum fault tolerance, where the flat connection microscopically certifies the fault tolerance property. In the present work, this quest is paired with an operator-algebraic framework for quantum error correction, originally developed by Kribs et al. and Beny et al. in refs.~\cite{kribs2006operatorquantumerrorcorrection,beny2007generalization}. This formulation was motivated by the work of Kitaev in ref.~\cite{kitaev2025almostidempotentquantumchannelsapproximate}. Holonomy and channel-induced transport of subspaces have also been studied in literature on open quantum systems, such as in refs.~\cite{Albert_2016,Kult_2008}, which can be understood as related constructions with different permitted evolutions and frame choices.
Our treatment of fault tolerance for GKP codes builds on previous work by the author in refs.~\cite{conrad2024latticesgatescurvesgkp,Burchards_2025}. Recent independent work by Mayrand and Royer in ref.~\cite{mayrand2026complex} has also developed a thorough dictionary between the mathematics of polarized abelian varieties and GKP codes and their code states, including the Jacobian/passive-Clifford correspondence.
The final section, on magic states in GKP codes, builds on and extends observations made in previous work by the author~\cite{conrad2024latticesgatescurvesgkp}.
The mathematical tools used here are largely standard. For background, we refer to refs.~\cite{KitaevShenVyalyi,Holevo2019,Paulsen_2003,NielsenChuang2010} on operator theory, refs.~\cite{Birkenhake_2004,Tata_1} on complex algebraic and analytic methods, and ref.~\cite{Hall2013} on the relevant physics.

\subsection*{Acknowledgments}
This project has benefitted from discussions with many colleagues over its course. In particular, I thank Victor Albert for first bringing the question of fiber-bundle fault tolerance to my attention. I thank Daniel Gottesman and Lucy Zhang for explaining their motivation behind ref.~\cite{gottesman2017fibre} and for helpful discussions on this work, and I am grateful to Julio Magdalena for regular discussions on formulations of fault tolerance, as well as feedback on this manuscript. I further thank Steve Flammia and Thomas Vidick for discussions on this topic.
During the later stages of this project, I have made use of LLMs (ChatGPT and Claude) for finding and understanding the literature, suggesting and optimizing proof strategies and reviewing the manuscript.
\newpage

\section{Knill-Laflamme geometry of fault tolerance}\label{sec:KL_Geometry}
In this section we define a generalization of the Knill-Laflamme conditions~\cite{knill1997theory} to compare multiple error channels and different encodings, and show how it defines a geometric structure on the space of encoded logical operator algebras associated to a Hilbert space. We show that this perspective allows us to develop a formulation of fiber bundle fault tolerance for general quantum channels. The framework we develop in the following generalizes the idea of fiber bundle fault tolerance by Gottesman and Zhang~\cite{gottesman2017fibre} to the setting of operator algebra quantum error correction initially investigated by Kribs et al. in ref.~\cite{kribs2006operatorquantumerrorcorrection}, drawing inspiration from Kitaev's recent work in ref.~\cite{kitaev2025almostidempotentquantumchannelsapproximate}.

\subsection{Operator-algebra quantum error correction and fault tolerance}
Let $\CH$ be a separable Hilbert space and let
$A=\bigoplus_{k=1}^{K}B(K_k)$ be a finite-dimensional $C^*$-algebra,
equipped with the canonical trace $\tau_A=\sum_k\Tr_{K_k}$.
We use the inner products
\begin{equation}
\langle a,b\rangle_A=\tau_A(a^*b),\qquad
\langle X,Y\rangle_{\mathrm{HS}}=\Tr_{\CH}(X^*Y)
\end{equation}
to identify the space of Hilbert–Schmidt (HS) operators $X,Y\in S_2(\CH)$ and consider normal completely positive (CP) maps $\CE$ on the space of bounded operators $B(\CH)$ with dual $\CE^*$ taken under trace duality. A CP map is a quantum channel if it is unital, $\CE(I)=I$, and we will generally not make this assumption. 
A physical state is a positive, trace-one element $\rho\in S_1(\CH)=\{X\in B(\CH),\, \Tr|X|<\infty\}$ that defines a bounded linear functional 
\begin{equation}
    M_\rho:\; B(\CH) \to \C,\quad  O\mapsto \Tr[\rho O].
\end{equation}

We focus on the description of operator algebras associated to quantum error correcting codes, in a generalized sense. Consider $B(\CH)$, the space of bounded operators on a Hilbert space $\CH$. A quantum error correcting code defines a subspace $\Phi B(\CH)\subset B(\CH)$ via an idempotent map $\Phi=\Phi^2=\Phi^*$ that represents the $C^*$ algebra of logical operators described by $A$. Via this description, we can investigate the evolution of encoded quantum states purely through the evolution of the subspace defined by the idempotent $\Phi$. We make this precise in the following.

\begin{definition}[Framed Grassmannians]
    The framed Grassmannian $\overline{\Gr}_A\lr{\CH}$ on $\CH$ for $A$ is the collection of triples
        $(\Phi_x, \Delta_x, \Lambda_x), x\in \FX$,
    where $\FX$ is a label set and
    \begin{equation}
\Phi_x:B(\CH)\to B(\CH),\; \Phi_x=\Phi_x^*=\Phi_x^2,
\end{equation}
are completely positive (CP) idempotents
together with a frame at $\Phi_x \in\Gr_{A}(\CH)$ provided by a choice of CP maps
\begin{equation}
\Delta_x:A\to B(\CH),
\qquad
\Lambda_x:B(\CH)\to A
\end{equation}
such that
\begin{equation}
\Lambda_x\Delta_x=\id_A,
\qquad
\Delta_x\Lambda_x=\Phi_x.\label{eq:frame_cond}
\end{equation}
A frame $(\Delta_x, \Lambda_x)$ is normalized when $\Lambda_x=\Delta_x^*$. If, in addition, $\Delta(1_A)^2=\Delta(1_A)$, we call $\Phi$ a subspace framed (SF) idempotent.  
\end{definition}
%
Following ref.~\cite{kitaev2025almostidempotentquantumchannelsapproximate}, the framed idempotent structure induces a representation of $A$ on ${\rm Im}\Phi$ via the Choi-Effros product~\cite{ChoiEffros}. 
\begin{definition}
    Let $(\Phi, \Delta, \Lambda)\in \overline{\rm Gr}_{A}(\CH)$ be a normalized framed idempotent with $\Lambda=\Delta^*$ and let $1_A\in A$ be the unit element. For $X,Y\in\Im\Phi$, the weighted Choi-Effros product is given by
    \begin{equation}
        X \times_{\Delta} Y = \Phi\lr{X\Delta(1_A)^{-1}Y}.
    \end{equation}
\end{definition}
The expression $X\Delta(1_A)^{-1}Y$ is understood as a bounded strong limit~\cite[eq.~(1)]{LiebRuskai}.
We can check that, for a framed idempotent, the product defined above makes $\Delta$ a $*$-homomorphism from the algebra. Each choice of normalized frame therefore realizes a representation of the algebra $A$ on $\Im\Phi$. 
\begin{lemma}
    Let $(\Phi, \Delta, \Lambda)\in \overline{\rm Gr}_{A}(\CH)$ be a normalized framed idempotent with $\Lambda=\Delta^*$, let $1_A$ be the unit of $A$ and let $a, b\in A$. $\Delta$ realizes a $*$-homomorphism 
    \begin{equation}
        \Delta(ab)=\Delta(a) \times_{\Delta} \Delta(b).
    \end{equation}
\end{lemma}
\begin{proof}
 Complete positivity of $\Psi = \Delta, \Lambda$ implies, via Stinespring dilation~\cite{Steinspring}, that we can write $\Psi(x)=S^*\pi(x)S$ with $\pi$ a $*$-representation of the domain of $\Psi$ and $S$ a bounded linear operator. Lieb and Ruskai~\cite{LiebRuskai} showed that\footnote{In ref.~\cite{LiebRuskai}, $C>0$ denotes positive semidefiniteness, i.e., $\langle x,Cx\rangle\geq0$ for all $x\in\CH$.} %
 %
 %
 %
 any such map satisfies $\Psi(u^*u)\geq \Psi(u^*v)\Psi(v^*v)^{-1}\Psi(v^*u)$. Consequently, taking $\Psi=\Delta$, $u=a$, $v=1_A$ and $\Psi=\Lambda, u=\Delta(1_A)^{-1/2}\Delta(a)$, $v=\Delta(1_A)^{1/2}$, one obtains the two inequalities
   \begin{equation}
    a^*a\leq \Lambda\lr{\Delta(a^*)\Delta(1_A)^{-1}\Delta(a)}\leq a^*a,
   \end{equation}
   and hence equality. Setting $a\mapsto a+b$ and $a\mapsto a+ib$, comparing terms using linearity, and finally applying $\Delta$ proves the statement.
\end{proof}

In the special case of SF-idempotents, $\Delta(1_A)^2=\Delta(1_A)$ and the encoded unit is a projector. This is the case when the encoder satisfies $\Delta(ab)=\Delta(a)\Delta(b)$ under the standard product. SF-idempotents realize a subspace described via
\begin{equation}
    \Phi_x = \sum_{k=1}^K \Pi_{x,k}\otimes \overline{\Pi}_{x,k},\quad \Pi_{x,k}=V_{x,k}V_{x,k}^{\dagger},
\end{equation}
with isometries $V_{x,k}$ with $V_{x,k}^{\dagger}V_{x,k}=I_{K_k}$.
SF-idempotents need not be unital on $B(\CH)$. Rather, the definition implies that
\begin{equation}
    \Phi_x(I) = \Delta_x(1_A) = \sum_{k=1}^K\Pi_{x,k} \eqqcolon \Pi_x. 
\end{equation}
We also observe that to each idempotent $\Phi_x$, one may associate multiple frames, as for an invertible element $g\in \Aut(A)$ we have
\begin{equation}
    (\Delta_x g)(g^{-1}  \Lambda_x) = \Delta_x \Lambda_x=\Phi_x.
\end{equation}
Throughout this work we assume $g\in \Aut(A)$ to respect the $*$-structure of the operator algebra.
To emphasize this freedom, we define a frame bundle over the unframed Grassmannian.
\begin{definition}[Grassmannian frame bundle]
   Let  $\overline{\Gr}_A\lr{\CH}$ be the set of normalized framed Grassmannians on $\CH$ for $A$. The Grassmannian frame bundle is defined by forgetting the frame
   \begin{equation}
        \pi:\, \overline{\Gr}_A\lr{\CH} \to \Gr_A\lr{\CH},\quad (\Phi_x, \Delta_x, \Lambda_x) \mapsto \Phi_x,\quad  x\in \FX.
   \end{equation}
\end{definition}

We equip the induced space of idempotents $\Gr_A(\CH)$ with the topology given by the completely bounded (cb) norm defined for maps between two $C^*$-algebras $\Lambda: A'\to A''$  as
\begin{equation}
    \|\Lambda\|_{\rm cb}=\sup_n \sup_{\substack{X \in M_n\otimes A'\\ X\neq 0}}\frac{\|(I_n\otimes \Lambda)(X)\|}{\|X\|},
\end{equation}
where $M_n=B(\C^n)$ is the algebra of $n\times n$ matrices. The induced metric \begin{equation}d_{\rm cb}(\Psi, \Phi) \eqqcolon \|\Psi -\Phi\|_{\rm cb},
\end{equation}
defines the norm topology on $\Gr_A\lr{\CH}$: a subset $U\subset \Gr_A\lr{\CH}$ is called open if for all $\Psi\in U$ there is an $\epsilon > 0$ such that 
\begin{equation}
    d_{\rm cb}(\Psi, \Phi)<\epsilon \Rightarrow \Phi \in U.
\end{equation}
Let $U_{\alpha}, U_{\beta}$ be two open subsets such that $\Phi\in U_{\alpha}\cap U_{\beta}$, and for $\gamma \in \lrc{\alpha, \beta}$, let the normalized section $s_{\gamma}(\Phi)=(\Phi, \Delta_{\gamma}(\Phi),\Lambda_{\gamma}(\Phi))$ provide the local trivialization on $U_{\gamma}$. The previous discussion implies that there are elements $g_{\alpha\beta}(\Phi)=g_{\beta\alpha}(\Phi)^{-1}\in \Aut(A)$ such that
\begin{equation}
(\Delta_{\beta}(\Phi),\Lambda_{\beta}(\Phi))=(\Delta_{\alpha}(\Phi)g_{\alpha\beta}(\Phi),g_{\alpha\beta}(\Phi)^{-1}\Lambda_{\alpha}(\Phi)).
\end{equation}

\subsection{Correctable Grassmannians and correctable transport}
A classical approach in quantum error correction is to fix the quantum error correcting code (here specified through the idempotent $\Phi$) and examine error channels predominantly supported on correctable sets. Plainly, a correctable error channel can be defined as follows.
\begin{definition}[Correctable channel]\label{def:correctable}
A channel $\CE$ is correctable for a code point \(\Phi\) if there exists a recovery channel $\CR$ such that \begin{equation} \CR\mathcal E\Phi = c \Phi \label{eq:correctable}\end{equation} for a constant $c\in \R_+$.
\end{definition} 
The requirement here is the existence of a projective inverse for $\CE$ on the image of $\Phi$. A channel is a completely positive trace-preserving (CPTP) map on $S_1(\CH)$.

\begin{definition}[Knill--Laflamme conditions]\label{def:KL}
    Let $\CE_s, \CE_t$ be a pair of quantum channels and let $(\Phi_{x}, \Delta_x, \Lambda_x), (\Phi_{y}, \Delta_y, \Lambda_y) \in \overline{\Gr}_A\lr{\CH}$ be a pair of framed idempotents. The Knill-Laflamme conditions are satisfied if there exist scalars $\alpha^k_{xy, st}\in \R_{>0}$ such that
    {\renewcommand{\theHequation}{KL}%
    \begin{equation}
        (\CE_s\Phi_x)^*\CE_t \Phi_y =\sum_{k=1}^K \alpha^k_{xy, st} \Delta_x p_{k}\Lambda_y.\tag{KL}\label{eq:KL}
    \end{equation}}
\end{definition}

We first examine what this condition implies for $x=y, s=t$ in a simplified setting. To this end, we define, for a framed idempotent $(\Phi_x, \Delta_x, \Lambda_x)$ and $\bs{a}\in \R_{>0}^K$, the diagonal map
\begin{equation}
    \mathfrak{m}_{\bs{a}}=\sum_k a_k p_k ,\quad   D_{x,\bs{a}}= \Delta_x  \mathfrak{m}_{\bs{a}}\Lambda_x, 
\end{equation}
with $p_k$ the projection onto $B(K_k)$.
For an SF-idempotent $(\Phi_{x}, \Delta_x, \Lambda_x)\in \overline{\Gr}_A\lr{\CH}$ with $\Phi_x=\sum_{k}\Phi_{x,k}$ for $\Phi_{x,k}=\Pi_{x,k}\otimes \overline{\Pi}_{x,k}$ \eqref{eq:KL} requires the existence of real scalars $a_k\in\R^+$ such that 
\begin{equation}
    \Phi_{x}\CE^*\CE\Phi_{x} =D_{x,\bs{a}}>0.
\end{equation}
Define $\CW_x=\CE D_{x,\bs{a}}^{-\frac{1}{2}}$. This defines an isometry with
\begin{equation}
    \CW_x^*\CW_x=\Phi_x,\quad  \CW_x\CW_x^*=\CE D_{x,\bs{a}}^{-1}\CE^* \eqqcolon \Phi^{\CE}_x=\Phi^{\CE}_x, 
\end{equation}
which implements isometric transport between the initial code idempotent $\Phi_x$ and the error idempotent $\Phi^{\CE}_x=(\Phi^{\CE}_x)^2$ as well as their frames. Note, however, that this isometry need not preserve the SF property of the framed idempotents. %
For a fixed code $\Phi_x=\Pi\otimes\overline{\Pi}$ with code space projector $\Pi$, the above condition is equivalent to the usual Knill-Laflamme conditions, discussed, e.g., in ref.~\cite{NielsenChuang2010}. 

\begin{lemma}[Equivalence to subspace KL conditions]
    Let $\Phi=\sum_{k=1}^K\Pi_k\otimes\overline{\Pi}_k$ be an SF-idempotent with code space projector $\Pi$, and let 
    \begin{equation}
        \CE=\CE_s=\CE_t=\sum_{i} E_i\otimes \overline{E}_i,
    \end{equation}
    be specified by Kraus operators  $\lrc{E_i}_i$. Then eq.~\eqref{eq:KL} with $x=y$ and $s=t$ is equivalent to the condition that there exist matrices $\alpha^{k}=(\alpha^k)^{\dagger} $,
necessarily positive semidefinite with $\sum_i\alpha^k_{ii}=1$ such that
    \begin{equation}
        \Pi_k E_{i}^{\dagger}E_j\Pi_l = \delta_{kl}\alpha^{k}_{ij}\Pi_l.\label{eq:class_KL}
    \end{equation}
\end{lemma}
\begin{proof}
We begin with the reverse direction. Assuming eq.~\eqref{eq:class_KL}, we have
\begin{align}
    \Phi\CE^*\CE\Phi&=\sum_{kl}\sum_{ij} \Pi_k E_{i}^{\dagger}E_j\Pi_l \otimes \overline{\Pi_k E_{i}^{\dagger}E_j\Pi_l}\\ 
    &= \sum_{k}\sum_{ij} |\alpha^{k}_{ij}|^2 \Phi_k.
\end{align}
This satisfies eq.~\eqref{eq:KL}.  
For the forward implication, eq.~\eqref{eq:KL} gives the following identity for $\CE$ with Kraus operators $\lrc{E_i}_i$ and $A_{ij}^{kl}=\Pi_k E_i^{\dagger}E_j\Pi_l $:
\begin{equation}
    \sum_{ij}A_{ij}^{kl} \otimes \overline{A_{ij}^{kl}} = \delta_{kl} a_k \Pi_k\otimes \overline{\Pi}_k.
\end{equation}
For each pair $k=l$, the r.h.s. is a CP map of Kraus rank $1$, hence $A_{ij}^{kl} = \alpha^{kl}_{ij} \Pi_k$ for some $\alpha^{kl}_{ij} \in \C$. For $k\neq l$, a zero CP map has every Kraus operator zero,
so $A_{ij}^{kl}=0$. For each $k$, write $\alpha^k_{ij}=\alpha^{kk}_{ij}$. Then $\alpha^k$ is a Hermitian matrix with $\|\alpha^k\|^2_F=a_k$.
For a state $\ket{v}\in\Im\Pi_k$ and any finitely supported
coefficient vector $c$,
\begin{equation}
 c^\dagger\alpha^kc=\|\sum_i c_iE_i\ket{v}\|^2\geq0,
 \qquad
 \sum_i\alpha^k_{ii}
 = \braket{v|\sum_iE_i^\dagger E_i|v}=1.
\end{equation}
Thus $0<a_k=\Tr[(\alpha^k)^2]\leq1$. Since physical channels have a countable number of Kraus operators, Kraus sums converge on trace-class inputs in trace norm, and $\sum_iE_i^\dagger E_i=I$ converges strongly.
\end{proof}
The constant $a_k=\|\alpha^k\|^2_F$ can be interpreted as the purity of the quantum state over the different syndrome subspaces (see, e.g., ref.~\cite{Du2025} for a discussion of this interpretation).

In the simplest case, the assumption $\Phi\CE^*\CE\Phi=c\Phi$ implies the existence of a map $\CR_0=\Phi\CE^*$ satisfying eq.~\eqref{eq:correctable}. This definition, however, does not yield a valid quantum channel, which is required to be TP. We repair this using the following construction, which gives the classical transpose channel~\cite{Ng_2010} $\CR_T=\CR_0\CN$ on the error support.

\begin{lemma}[Normalization of HS isometries]
    Let $\Phi$ be an SF-idempotent with $\Pi=\Phi(I)$ and let $\CE$ be a channel such that $\CE(\Pi)\geq \epsilon Q >0$ for some $\epsilon>0$, where $Q=P_{{\rm supp} (\CE(\Pi))}$.  Then, with $\Pi=\sum_{k=1}^K\Pi_k$, the KL condition $\Phi\CE^*\CE\Phi=c\Phi$ implies that
\begin{equation}
\CN: B(\CH)\to B(\CH),\quad \CN(X)=\CE^{-\frac{1}{2}}(\Pi)\vert_QX\CE^{-\frac{1}{2}}(\Pi)\vert_Q,
\end{equation}
is a CP map with $\CN \CE \Phi (I)=\CN (\CE(\Pi))=Q$. Moreover, the transpose recovery $\CR_T=\Phi\CE^*\CN$ satisfies $\CR_T\CE\Phi=\Phi$.
\end{lemma}
\begin{proof}
    Denote by $\CE^{-1/2}(\Pi)$ the inverse square root of $\CE(\Pi)$ on its support $Q$ and $0$ elsewhere. Since $\CE^{-1/2}(\Pi)$ is self-adjoint, $\CN$ is CP. Trace preservation on $Q$ and the recovery identity follows from direct computation using the KL conditions and $\Phi(I)=\Pi$ (see also ref.~\cite{Ng_2010}).
\end{proof}

The normalization implements a CPTP recovery on the support $Q$, and can be lifted to a full CPTP map $ \CR(X)=\CR_T(X)+\Tr[(I-Q)X]\rho_0$ for any auxiliary density operator $\rho_0$. 
A crucial aspect of the KL condition is that it enables isometric transport of frames, not just between a reference framed idempotent and its error image, but also between different idempotents. Following the same proof strategy as before, for $\CE_s=\CE_t$ the cross KL condition implies nonnegative block-scalar overlaps between normalized code frames.  Strict positivity of these input overlaps will additionally
be required to identify their full frame transport with a polar factor.
\begin{lemma}[Cross KL]\label{lem:crossKL}
    Let $\CE$ be a quantum channel and let $(\Phi_{x}, \Delta_x, \Lambda_x),$ $(\Phi_{y}, \Delta_y, \Lambda_y) \in \overline{\Gr}_A\lr{\CH}$ be a pair of normalized framed idempotents. If there exists $\bs{a}\in\R_{\geq 0}^K$ such that
    \begin{equation}
        (\CE\Delta_x)^*\CE\Delta_y = \mathfrak{m}_{\bs{a}},
    \end{equation}
    there exists $\bs{b}\in\R_{\geq 0}^K$ such that
    \begin{equation}
        \Lambda_x\Delta_y=\mathfrak{m}_{\bs{b}},\quad \bs{b}\in \R_{\geq 0}^K.
    \end{equation}
    Moreover, if $\epsilon \coloneqq\|\Phi_x-\Phi_y\|_{\rm cb}<1$, it holds that $b_k\geq \sqrt{1-\epsilon^2}>0\;\forall k=1,\hdots,K$.
\end{lemma}
\begin{proof}
    This proof follows the same strategy as the previous one. By assumption, $\Delta_x, \Delta_y$ are CP, thus there exist Kraus operators $M_{k,i},\, N_{k,j}:\, K_k\to \CH$ such that 
    \begin{equation}
        \sum_{i,j,k, l, \alpha,\beta} M_{k,i}^{\dagger}E_{\alpha}^{\dagger}E_{\beta}N_{l,j} \otimes \overline{M_{k,i}^{\dagger}E_{\alpha}^{\dagger}E_{\beta}N_{l,j}}=\sum_{k} a_k p_k.
    \end{equation}
    Again, this expression implies that on the l.h.s., for $k\neq l$ we have $M_{k,i}^{\dagger}E_{\alpha}^{\dagger}E_{\beta}N_{l,j}=0$. For $k=l$, comparing Kraus ranks shows that there are $c_{ij,\alpha\beta}^k\in \C$ such that
    \begin{equation}
    M_{k,i}^{\dagger}E_{\alpha}^{\dagger}E_{\beta}N_{l,j} = \delta_{kl}c_{ij,\alpha\beta}^k 1_{K_k},\quad a_k=\sum_{i,j,\alpha,\beta}|c_{ij,\alpha\beta}^k|^2. \label{eq:crossKL_cond_proof}
    \end{equation}
If $a_k=0$, all the corresponding Kraus operators vanish and we take $c_{ij,\alpha\beta}^k=0$.
Summing over the diagonal indices $\alpha=\beta$ in eq.~\eqref{eq:crossKL_cond_proof} using $\sum_{\alpha}E^{\dagger}_{\alpha}E_{\alpha}=I$ then gives
    \begin{equation}
        M_{k,i}^{\dagger}N_{l,j}=\delta_{kl}\lrq{\sum_{\alpha}c_{ij,\alpha\alpha}^k}1_{K_k},
    \end{equation}
    and hence
    \begin{align}
        \Lambda_x\Delta_y&=\sum_{i,j,k,l} M_{k,i}^{\dagger}N_{l,j}\otimes \overline{M_{k,i}^{\dagger}N_{l,j}}\\
        &=\sum_k b_k p_k,\quad b_k=\sum_{i,j}\left\vert\sum_{\alpha}c_{ij,\alpha\alpha}^k \right\vert^2.
    \end{align}
    Finally, let $Y=\Delta_y u,\; u\in A$. We have
    \begin{align}
        (\Phi_x-\Phi_y)^2Y&=(\Phi_x - \Phi_x\Phi_y-\Phi_y\Phi_x+\Phi_y )\Delta_y u \nonumber \\ 
        &=\Delta_y (\id_A-\mathfrak{m}_{\bs{b}}^2)u. 
    \end{align}
    Set $u\in \Im p_k$. Then $(\Phi_x-\Phi_y)^2Y = (1-b_k^2)Y$. Taking norms yields
    \begin{equation}
        1-b_k^2 \leq \|(\Phi_x-\Phi_y)^2\|_{\rm cb} \leq \epsilon^2,
    \end{equation}
    which implies the final inequality.
\end{proof}

\begin{theorem}[Geometric KL transport]\label{thm:GeoKL}
    Let $\CE_s, \CE_t$ be a pair of quantum channels acting on $S_1(\CH)$ and consider a pair of framed idempotents $(\Phi_{x}, \Delta_x, \Lambda_x)$, $(\Phi_{y}, \Delta_y, \Lambda_y)$ such that $\Phi_{x}, \Phi_{y}$ restrict to orthogonal projections on $S_2(\CH)$ and have ranges in $S_1(\CH)$. Assume there exist $\bs{a}, \bs{b}\in \R_{>0}^K$ and $\gamma\in \R_{>0}$ such that
    \begin{align}
        (\CE_s\Phi_x)^*\CE_s\Phi_x = D_{x,\bs{a}},\quad (\CE_t\Phi_y)^*\CE_t\Phi_y = D_{y,\bs{b}}\label{eq:xy_KL_pair}
    \end{align} 
    and
    \begin{equation}
        (\CE_s\Phi_x)^*\CE_t\Phi_y = \gamma \Delta_x\Lambda_y,\quad  (\CE_t\Phi_y)^*\CE_s\Phi_x = \gamma \Delta_y\Lambda_x.\label{eq:xy_transport_cond}
    \end{equation} 
    
Let $\CW_{x}=\CE_sD_{x,\bs{a}}^{-\frac{1}{2}}, \CW_{y}=\CE_t D_{y,\bs{b}}^{-\frac{1}{2}}$. For $z=x,y$ and $Q_{z}=\CW_{z}\CW_{z}^*$, the map 
\begin{equation}\CW_{z}\cdot (\Phi_{z}, \Delta_z, \Lambda_z)=(Q_{z}, \CW_{z}\Delta_z, \Lambda_z \CW_{z}^*)\label{eq:W_map}
\end{equation}
defines isometries between $\operatorname{Ran}\Phi_z$ and $ \operatorname{Ran}Q_z$ such that the diagram

    \begin{equation}
\begin{tikzcd}[column sep=large, row sep=huge]
& A
  \arrow[dl, "\Delta_y"']
  \arrow[dr, "\Delta_x"]
&
\\
\operatorname{Ran}\Phi_y
  \arrow[rr, "{U_{x\leftarrow y}=\Delta_x\Lambda_y}"]
  \arrow[d, bend right=45, "\CW_y"']
&&
\operatorname{Ran}\Phi_x
  \arrow[d, bend left=45, "\CW_x"]
\\
\operatorname{Ran}Q_y
  \arrow[rr, "{J_{x\leftarrow y}}"']
  \arrow[u, bend right=45, "\CW_y^*"']
&&
\operatorname{Ran}Q_x
  \arrow[u, bend left=45, "\CW_x^*"]
\end{tikzcd}
\end{equation}
commutes, where, with $|A|=(A^*A)^{\frac{1}{2}}$ and inverse $|A|^{-1}$ defined to act as $0$ on the kernel of $A$,
\begin{equation}
    J_{x\leftarrow y} = \CW_x\Delta_x\Lambda_y\CW_y^*=Q_xQ_y|Q_xQ_y|^{-1}. \label{eq:J_geom_Jfr}
\end{equation}
\end{theorem}
\begin{proof}
Put $U_{x\leftarrow y}=\Delta_x\Lambda_y$. The conditions imply $U_{x\leftarrow y}^*=U_{y\leftarrow x}$ and further $U_{x\leftarrow y}^*U_{x\leftarrow y}=\Phi_y,\;U_{x\leftarrow y}U_{x\leftarrow y}^*=\Phi_x$.
Thus $U$ is a partial isometry. Eq.~\eqref{eq:xy_KL_pair} likewise yields
$\CW_z^*\CW_z=\Phi_z$ for $z=x,y$. The supported inverse powers of the code Gram operators are positive central block multipliers in the corresponding frames so $D_{\bs{a}}$, $\CW_z$, and the transported frame maps are CP.
By evaluating the defining equations, we can check that the r.h.s. of eq.~\eqref{eq:W_map} yields a valid framed idempotent and that $\CW_{z}$ defines partial isometries between $\operatorname{Ran}\Phi_z$ and $ \operatorname{Ran}Q_z$. Similarly, direct computation verifies that the lower square in the diagram commutes. The crucial non-immediate step is to show that eq.~\eqref{eq:J_geom_Jfr} holds, which states that the isometry of frames is equal to the geometric transport via projection $Q_xQ_y$ between different frames. 
We first compute
\begin{equation}
    (Q_xQ_y)^*Q_xQ_y = Q_yQ_xQ_y = \gamma^2 \CW_y D^{-1}_{y, \bs{a} \cdot \bs{b}} \CW_y^*,
\end{equation}
where $D^{-1}_{y, \bs{a} \cdot \bs{b}}=D^{-1}_{y, \bs{a}}D^{-1}_{y, \bs{b}}$. Spectral calculus then gives
\begin{equation}
 |Q_xQ_y|^{-1}=\gamma^{-1}\CW_y
 D_{y,(\bs a\cdot\bs b)^{1/2}}\CW_y^*.
\end{equation}
Here $D_{y,\bs a}=U^*D_{x,\bs a}U$ is positive on $\Im\Phi_y$, commutes with $D_{y,\bs b}$, and all inverses vanish on the complementary supports. %
Furthermore, we compute
    \begin{equation}
        Q_xQ_y=\gamma \CW_xD_{x, \bs a}^{-\frac{1}{2}}\Delta_x \Lambda_y D_{y, \bs b}^{-\frac{1}{2}}\CW_y^*=\gamma \CW_x\Delta_x \Lambda_y D_{y, \bs a\cdot \bs{b}}^{-\frac{1}{2}}\CW_y^*,
    \end{equation}
    and, combining these expressions, we obtain
    \begin{equation}
Q_xQ_y|Q_xQ_y|^{-1}=\CW_x\Delta_x\Lambda_y
 D_{y,\bs a\cdot\bs b}^{-\frac12}
 D_{y,(\bs a\cdot\bs b)^{1/2}}\CW_y^*
 =\CW_x\Delta_x\Lambda_y\CW_y^*.
    \end{equation}
\end{proof}
The key observation of this theorem is that the transport condition in eq.~\eqref{eq:xy_transport_cond} ensures that the isometric transport through the error idempotent, after applying the inverse partial isometry, faithfully implements the change of frame $U_{x\leftarrow y}=\Delta_x\Lambda_y$. For $\CE_s=\CE_t$ on normalized frames, this theorem and lemma~\ref{lem:crossKL}
together imply, provided the resulting input overlaps $\Lambda_x\Delta_y$
are strictly positive on every block, that the error isometries intertwine geometric transport between points $x$ and $y$ with geometric transport on their error images. In quantum error correction terms, the general theorem can be interpreted as a Knill-Laflamme condition for code switching. 

To study continuous fault-tolerant protocols, we next construct smooth manifolds from allowed control operations and impose an infinitesimal compatibility condition ensuring that their induced frame transport is horizontal.

\subsection{Smooth correctable spaces and flat connections}
The previous discussion motivates the definition of a correctable subgrassmannian using the conditions \eqref{eq:KL}. In the following we will show, using these conditions, that a subgrassmannian can be constructed that has a naturally flat connection. For notational convenience, we will assume $\Lambda_x=\Delta_x^*$ for every framed idempotent $(\Phi_x, \Delta_x, \Lambda_x)\in\overline{\Gr}_A(\CH)$  in the following.

Let $\FE=\lrc{\CE_1,\hdots,\CE_N}$ be a set of quantum channels acting on $S_1(\CH)$, which we regard as the collection of all accessible channels for the quantum system at hand. We define correctable subsets and neighbourhoods on $\Gr_A(\CH)$.
\begin{definition}[$S\subseteq \FE$-correctable subsets]\label{def:ScorrSub}
    Let $S$ be a set of quantum channels acting on $S_1(\CH)$. A family $W\subset\overline{\Gr}_A(\CH)$ of normalized framed idempotents $(\Phi_{x}, \Delta_x, \Lambda_x=\Delta_x^*)\in \overline{\Gr}_A(\CH)$ is jointly $S$-correctable if, for every pair $x,y$ and every $\CE_s, \CE_t \in S$,
    \begin{equation}
        (\CE_s\Delta_x)^*(\CE_t\Delta_y)=\mathfrak{m}_{\bs{b}_{xy, st}},\quad \bs{b}_{xy, st}\in \R_{>0}^K.\label{eq:xyKL_frames}
    \end{equation}
    We write $\bs{a}_{x,s}=\bs{b}_{xx, ss}$.
\end{definition}
\begin{definition}[$S$-correctable framed neighbourhood]\label{def:ScorrNeigh}
Consider $M\subset \Gr_A(\CH)$ with its relative $\rm cb$-topology. An $S$-correctable framed neighbourhood is a relatively open $U\subset M$ for which there is a continuous section of normalized frames $\overline{s}(\Phi_x)=(\Delta_{\Phi_x},\Lambda_{\Phi_x}=\Delta^*_{\Phi_x} )$ over $U$ such that $(\Phi_x,\overline{s}(\Phi_x))_{\Phi_x \in U}$ is jointly $S$-correctable. If $M$ has a smooth structure, $\overline{s}$ is also required to be smooth.
\end{definition}

Using Lemma~\ref{lem:crossKL}, we show that changes between continuous local KL frames are locally constant on overlaps of open subsets. Choices of frames (indexed by $\alpha, \beta$ in the following) and code points (indexed by $x, y$) are ``glued together'' through the KL transport rules on any $M\subset\Gr_A(\CH)$ admitting such an open cover. This implies that the normalized frame bundle restricted to $M$ is topologically flat.

\begin{theorem}[Flat transport from local KL]\label{thm:ctg_topological}
Let $M\subset\Gr_A(\CH)$, with its relative $\rm cb$ topology, admit an open cover by $S_\alpha$-correctable framed neighbourhoods $U_\alpha$,  for non-empty sets of quantum channels $\emptyset \neq S_\alpha\subseteq\FE$.
Then the local frame transitions are locally constant, and
\begin{equation}
 \mathcal T^\alpha_{y\leftarrow x}
   =\Delta_y^\alpha\Lambda_x^\alpha,
 \qquad \Phi_x,\Phi_y\in U_\alpha,
\end{equation}
defines isometric transport between code ranges along continuous paths in $M$, depending only on their endpoint-fixed homotopy classes. 
In particular,
a base frame at $\Phi_0\in M$ defines the monodromy representation
\begin{equation}
 \rho:\pi_1(M,\Phi_0)\longrightarrow\operatorname{Aut}(A),
 \qquad \rho([\gamma])=\Lambda_0\mathcal T_\gamma\Delta_0.
\end{equation}
\end{theorem}
\begin{proof}
    Choose any $\CE\in S_\alpha$. By Lemma~\ref{lem:crossKL}, the local KL condition with $s=t$ gives 
    \begin{equation}
         R_\alpha(x,y):=\Lambda_x^\alpha\Delta_y^\alpha
        =\mathfrak m_{\bs c^\alpha_{xy}}\geq0,
        \qquad R_\alpha(x,x)=\id_A.
    \end{equation}
    On an overlap, write $\Delta_x^\beta=\Delta_x^\alpha g(x)$, where $g(x)=\Lambda_x^\alpha\Delta_x^\beta\in\operatorname{Aut}_*(A)$ is HS-unitary. Then $R_\beta(x,y)=g(x)^*R_\alpha(x,y)g(y)$. For fixed $\Phi_x\in U_\alpha\cap U_\beta$, let $\Phi_y \in U_\alpha\cap U_\beta$ with $\|\Phi_x-\Phi_y\|_{\rm cb}<1$. Lemma~\ref{lem:crossKL} then makes both $R_\alpha(x,y)$ and $R_\beta(x,y)$ positive definite.
    Taking their polar decompositions gives $\id_A=g(x)^*g(y)$, hence $g(y)=g(x)$. Comparing sections from two admissible covers gives locally constant transition maps between sections from the two covers on their intersections, therefore both covers induce the same transport. Homotopy invariance
    \begin{equation}
        \mathcal T^\alpha_{z\leftarrow y}
        \mathcal T^\alpha_{y\leftarrow x}
        =\mathcal T^\alpha_{z\leftarrow x},
        \qquad \mathcal T^\alpha_{x\leftarrow x}=\Phi_x
    \end{equation}
    follows directly from the frame identities. Define $\CT_{\gamma}$ by composing these maps along a finite subdivision subordinate to the cover. The frame identities and locally constant transitions make this independent of the subdivision and neighbourhood choices.
    For an endpoint-fixed homotopy, compactness and openness provide a finite grid of its parameter square with each cell mapped into one neighbourhood, and the above equation implies homotopy invariance. Concatenation and reversal of paths give the stated representation.
\end{proof}

The preceding theorem requires a subspace $M\subset\Gr_A(\CH)$ together with a covering by $S_\alpha$-correctable neighbourhoods. We did not specify this subspace any further, and the assumption of $S_\alpha$-correctable neighbourhoods was the crucial ingredient in obtaining locally constant frame transitions. In the following we ``enrich'' this structure by studying convex combinations of quantum channels that jointly satisfy the KL conditions. We will show that this allows us to build subspaces $M$ that have a smooth, locally Euclidean structure.
We denote by
\begin{equation}
    \Diamond_d=\lrc{\bs{p}\in \R_{\geq 0}^{d},\; \sum_{s=1}^{d}p_s=1},\quad  \bigcirc_d=\lrc{\bs{p}\in \R_{>0}^{d},\; \sum_{s=1}^{d}p_s=1}
\end{equation}
the closed and open $(d-1)$-simplices, which characterize all $d$ outcome probability distributions. 
For a set $S$ of channels define 
\begin{equation}
\CE(\bs{p})=\sum_{s=1}^{|S|} p_s \CE_s,\quad \bs{p}\in \R^{|S|},
\end{equation}
and for a fixed $S$-correctable normalized frame $(\Phi, \Delta, \Lambda)\in \overline{\Gr}_A(\CH)$ write (for $c\in \R$)
\begin{align}
    X_{\Delta}(S)&\coloneqq {\rm span}_{\R}\lrc{\CE_s\Delta, \CE_s\in S}, \quad
    &&X^c_{\Delta}(S)\coloneqq \{\CE(\bs{p})\Delta,\; \sum_i p_i =c\}, \nonumber\\
    C_{\Delta}(S)&\coloneqq\lrc{\CE(\bs{p})\Delta,\; \bs{p}\in \Diamond_{|S|}}, \quad
    &&\CO_{\Delta}(S)\coloneqq \lrc{\CE(\bs{p})\Delta,\; \bs{p}\in \bigcirc_{|S|}}.
\end{align}

\begin{lemma}[Local manifold structure of mixtures]
\label{lem:mixture_manifold}
    Let $(\Phi, \Delta, \Lambda)\in \overline{\Gr}_A(\CH)$ be a normalized jointly $S$-correctable frame, where $\varnothing\neq S\subseteq\FE$ is finite and $\Delta(A)\subset S_1(\CH).$ We denote its KL transport under $\CE(\bs{p})$ by 
    \begin{equation}\CW_{\CE(\bs{p})}=\CE(\bs{p})(\Phi\CE(\bs{p})^*\CE(\bs{p})\Phi)^{-\frac{1}{2}}.\end{equation} 
    Then
    \begin{equation}
        \CO_{\Delta}(S)\to \Gr_A(\CH),\quad \CE(\bs{p})\Delta\mapsto \Phi^{\CE(\bs{p})} = \CW_{\CE(\bs{p})}\CW_{\CE(\bs{p})}^*
    \end{equation}
    is a smooth embedding into the Banach space of completely bounded maps $B(\CH)\to B(\CH)$, equipped with its cb norm, with image contained in $\Gr_A(\CH)$. Its image is a smooth embedded manifold of dimension $r=\dim_{\R}X_{\Delta}(S)-1$, with smooth normalized frame section 
    \begin{equation}
        \overline{s}(\Phi^{\CE(\bs{p})})=\lr{\CW_{\CE(\bs{p})}\Delta, \, \Lambda \CW_{\CE(\bs{p})}^*}.\label{eq:lem7_s}
    \end{equation}
    In particular, every point of this image has a neighbourhood smoothly parametrized by an open subset of $\R^r$.
\end{lemma}
\begin{proof}
We show that there exist smooth maps and subsets $U\subset \R^r$ such that for fixed $F_0\in \CO_{\Delta}(S)$,
\begin{equation}
    U\subset X_{\Delta}^0(S)\sim\R^r \xrightarrow{\bs{h}\mapsto F_0+\bs{h}} \CO_{\Delta}(S)\xrightarrow{F\mapsto \Phi^F} \Gr_{A}(\CH). \label{eq:lem7_over}
\end{equation}
 Trace preservation gives for $\bs v\in\R^{|S|},\ a\in A $
    \begin{equation}
        \Tr[\CE(\bs v)\Delta(a)]=\Bigl(\sum_s v_s\Bigr)\Tr[\Delta(a)].
    \end{equation}
    Since $\Tr\Delta(1_A)>0$, $X^1_\Delta(S)$ is an affine hyperplane in
    $X_\Delta(S)$ of dimension $r$.
    The affine map $\bigcirc_{|S|}\to X^1_\Delta(S)$,
    $\bs p\mapsto\CE(\bs p)\Delta$, is open, so $\CO_\Delta(S)$ is relatively
    open. 
    For any $\bs q,\bs p\in\Diamond_{|S|}$, the joint KL conditions yield
    \begin{equation}
        (\CE(\bs p)\Delta)^*\CE(\bs q)\Delta =\mathfrak m_{\sum_{s,t}p_sq_t\bs b_{st}},
        \quad  \bs b_{st}>0.\label{eq:62}
    \end{equation}
    It is hence immediately verified that
    \begin{equation}
        \lr{\CW_{\CE(\bs p)}\CW_{\CE(\bs p)}^*, \CW_{\CE(\bs p)}\Delta, \, \Lambda\CW_{\CE(\bs p)}^*}\in\overline{\Gr}_A(\CH)
    \end{equation}
    forms a valid framed idempotent.
For $\bs q=\bs p$ in the expression above, define the Gram coefficients
    \begin{equation}
        \bs a(\bs p):=\sum_{s,t}p_sp_t\bs b_{st},
        \qquad (\CE(\bs p)\Delta)^*\CE(\bs p)\Delta
        =\mathfrak m_{\bs a(\bs p)}.
    \end{equation}
    Each $a_k(\bs p)$ is a convex combination of the strictly positive
    $b_{st,k}$, so $a_k(\bs p)>0$ for every $\bs p\in\Diamond_{|S|}$.
    The expressions
    \begin{equation}
        \CW_{\CE(\bs p)}\Delta=\CE(\bs p)\Delta\mathfrak m_{\bs a(\bs p)}^{-1/2},
        \qquad
        \Phi^{\CE(\bs p)}=\CE(\bs p)\Delta\mathfrak m_{\bs a(\bs p)}^{-1}(\CE(\bs p)\Delta)^*\label{eq:65}
    \end{equation}
    depend only on the encoder $F=\CE(\bs p)\Delta$ and are
    smooth as functions of $F$ in the cb norms. Since $S$ and the number of logical blocks are finite, these expressions are finite linear combinations of fixed completely bounded maps with smooth scalar coefficients, and hence are smooth in the $\rm cb$ norms.
    The coefficients $\bs a(\bs p)>0$ also depend only on $F$, and positivity ensures that taking inverses and square roots are smooth operations. Since the image of $F$ consists of trace-class operators, the Hilbert–Schmidt adjoint of $F$ extends to a bounded normal CP map on $B(\CH)$.
    This shows that the second arrow in eq.~\eqref{eq:lem7_over} is smooth and that the transported normalized frame in eq.~\eqref{eq:lem7_s} depends smoothly on the encoder.

    We next show that the second arrow in eq.~\eqref{eq:lem7_over} is an injective immersion. To this end, we identify the tangent space of $\CO_{\Delta}(S)\to \Gr(\CH)$ with $X^0_\Delta(S)$, whose elements satisfy $\Tr h(a)=0$ for every $h\in X^0_\Delta(S)$ and $a\in A$. We first show that $\Phi^{\CE(\bs p)}h=h$ with $h\in X^0_\Delta(S)$
    forces $h=0$. Indeed, linearity of the KL conditions makes
    $(\CE(\bs p)\Delta)^*h$ block scalar, hence
    \begin{equation}
        h=\CE(\bs p)\Delta\mathfrak m_{\bs c}.
    \end{equation}
    Taking traces on each block unit gives $\bs c=0$ and therefore also $h=0$.
    
    To prove injectivity, suppose $\Phi^{\CE(\bs p)}=\Phi^{\CE(\bs q)}$.
    The common idempotent fixes both encoders, so their difference $h=(\CE(\bs q)-\CE(\bs p))\Delta\in X^0_\Delta(S)$ satisfies
 $\Phi^{\CE(\bs p)}h=h$. The preceding observation gives $h=0$, hence $\CE(\bs p)\Delta=\CE(\bs q)\Delta$.

    For $\bs p \in\bigcirc_{|S|}$ and $\bs v \in \R^{|S|}$ with $\sum_{s}v_s=0$, let
    $h=\CE(\bs v)\Delta\in X^0_\Delta(S)$. The weights
    $\bs p+t\bs v$ remain positive for sufficiently small $|t|$, and from
    \begin{align}
        \frac{d}{dt}\Phi^{\CE(\bs p+t\bs v)}\CE(\bs p+t\bs v)\Delta \big\vert_{t=0}
        &= \frac{d}{dt} \CE(\bs p+t\bs v)\Delta\big\vert_{t=0} \nonumber \\
        \Rightarrow \lr{\frac{\partial}{\partial\bs{v}}\Phi^{\CE(\bs{p})}}\CE(\bs{p})\Delta + \Phi^{\CE(\bs{p})}\lr{\frac{\partial}{\partial\bs{v}}\CE(\bs{p})\Delta} &= \lr{\frac{\partial}{\partial\bs{v}}\CE(\bs{p})\Delta},
    \end{align}
    we find from the previous argument that if $\frac{\partial}{\partial\bs{v}}\Phi^{\CE(\bs{p})}=0$, we also have $\frac{\partial}{\partial\bs{v}}\CE(\bs{p})\Delta=0$. 
    Hence the differential of $\CE(\bs p)\Delta\mapsto\Phi^{\CE(\bs p)}$ is injective on $X^0_\Delta(S)$, so this map is an immersion.
    
    Compactness of $C_\Delta(S)$ makes the continuous injection  $\CE(\bs p)\Delta\mapsto\Phi^{\CE(\bs p)}$ a homeomorphism onto its image. Its restriction to $\CO_\Delta(S)$ is therefore a smooth embedding.

    Translation identifies a neighbourhood of any point in $\CO_\Delta(S)$ with an open subset of $X^0_\Delta(S)\cong\R^r$, giving the claimed local parametrizations (the first arrow of eq.~\eqref{eq:lem7_over}).
    Composing the inverse of the second arrow, defined on its image, with the transported-frame map gives the claimed smooth normalized section.
\end{proof}

Lemma~\ref{lem:mixture_manifold} associates to a normalized jointly $S$-correctable framed idempotent \\
$(\Phi,\Delta,\Lambda)\in \overline{\Gr}_A(\CH)$ a smooth embedded subset
\begin{equation}
 \{\Phi^{\CE(\bs p)}:\bs p\in\bigcirc_{|S|}\}
 \subset\Gr_A(\CH),
\end{equation}
equipped with the smooth normalized frame section
in eq.~\eqref{eq:lem7_s}. The following corollary shows that when the corners of these subsets are correctable for a set $S'$, the transported frames again yield jointly ${\rm conv}(S')$-correctable neighbourhoods. This mechanism then allows us to successively construct a jointly correctable framed submanifold $M\subset \Gr_A(\CH)$.

\begin{cor}[Correctable mixture patches]
\label{cor:correctable_mixtures}
In Lemma~\ref{lem:mixture_manifold}, suppose the transported
vertex frames
\begin{equation}
 \{\CW_{\CE_s}\cdot(\Phi,\Delta,\Lambda):\CE_s\in S\}
\end{equation}
are jointly $S'$-correctable, where
$\varnothing\neq S'\subseteq\FE$. Then the entire transported
family
\begin{equation}
 \{\CW_{\CE(\bs p)}\cdot(\Phi,\Delta,\Lambda):
                  \bs p\in\Diamond_{|S|}\}
\end{equation}
is jointly $\operatorname{conv}(S')$-correctable.
\end{cor}
\begin{proof}
This follows from the same argument as in eq.~\eqref{eq:62}.
\end{proof}

To assemble the correctable patches of Lemma~\ref{lem:mixture_manifold} and
Corollary~\ref{cor:correctable_mixtures} into a smooth
submanifold $M\subset \Gr_A(\CH)$, their images must be relatively $\rm cb$-open in their union
and have a common dimension. 

Finally, now that we have shown how a locally smooth structure $M\subset \Gr_A(\CH)$ of jointly correctable idempotents and frame sections can be built using the KL conditions, we are ready to define a connection one-form on the Grassmannian frame bundle. To this end, consider the expression
\begin{equation}
\omega\in\Omega^1\!\left(
\overline{\Gr}_A(\CH)\big\vert_M;\mathfrak{aut}(A)
\right),
\qquad
\omega=\Lambda d\Delta,
\end{equation}
where $\overline{\Gr}_A(\CH)\big\vert_M$ denotes the restriction to a submanifold $M\subset \Gr_A(\CH)$. We first verify that this connection 1-form transforms in the expected way. Under a change of frame, $(\Delta, \Lambda)\mapsto (\Delta g, g^{-1}\Lambda)$ for $g\in\Aut(A)$, the definition implies
\begin{equation}
    \omega\mapsto g^{-1}\Lambda d(\Delta g)= g^{-1}\omega g + g^{-1}dg,
\end{equation}
which is the expected transformation behaviour. To unpack this further for SF idempotents, we have at $(\Phi, \Delta, \Lambda)\in \overline{\Gr}_A(\CH)$ for $a_k\in B(K_k)$, using $0=d(V_{k}^{\dagger}V_{k})=dV_{k}^{\dagger}V_{k}+V_{k}^{\dagger}dV_{k}$,
\begin{align}
    \Lambda d\Delta(a_k) &= \Lambda d\lr{\sum_{k=1}^K V_{k}a_k V^{\dagger}_{k}} \nonumber\\
    &= [V^{\dagger}_{k}dV_{k}, a_k],
\end{align}
and one can verify that the above expression $[V^{\dagger}_{k}dV_{k},\cdot]$ realizes an element of $\mathfrak{aut}(A)$.
We summarize the core construction of this section in the following theorem.
\begin{theorem}[Flat connection on correctable Grassmannian]
\label{thm:ctg_smooth}
Let $M=\bigcup_{\alpha}U_{\alpha}\subset\Gr_A(\CH)$ be a countable union of smooth $r$-dimensional patches $U_{\alpha}\subset \Gr_{A}(\CH)$
embedded in the space of $\cb$ maps $\lr{L_{\rm cb}(B(\CH)),\, \|\cdot \|_{\rm cb}}$. Suppose each $U_{\alpha}$ is relatively open in $M$ and carries a normalized jointly $S_{\alpha}$-correctable frame section for $\varnothing\neq S_{\alpha}\subseteq \FE$. 
\\
Then $M$ is a smooth embedded manifold and its normalized frame bundle
\begin{equation}
     \pi:\overline{\Gr}_A(\CH)|_M\longrightarrow M
\end{equation}
is a smooth principal $\Aut(A)$-bundle with flat connection
\begin{equation}
    \omega = \Lambda d\Delta.
\end{equation}
\end{theorem}
\begin{proof}
    The embedded patches $U_{\alpha}\subset \Gr_{A}(\CH)$ inherit compatible structures on their open overlaps. Since the relative $\rm cb$-topology is Hausdorff and the collection $\{U_{\alpha}\}_{\alpha}$ is countable, $M$ acquires the structure of a manifold. 
    On each patch $U_{\alpha}$, let $\Phi_0\in U_{\alpha}$ be a reference point with frame $(\Delta_0, \Lambda_0)$. By lemma~\ref{lem:crossKL}, for points $\Phi_x\in U_{\alpha}$ with $\|\Phi_x-\Phi_0\|_{\rm cb}<1$, the transitions satisfy
    \begin{equation}
        R_x=\Lambda_0\Delta_x=\Lambda_x\Delta_0 >0.
    \end{equation}
    With $G_x=\Lambda_0\Phi_x\Delta_0=R_x^2$, the expressions
    \begin{equation}
        \Delta_x=\Phi_x\Delta_0G_x^{-\frac{1}{2}},\quad \Lambda_x=G_x^{-\frac{1}{2}}\Lambda_0\Phi_x
    \end{equation}
    depend smoothly on $\Phi_x$ in $\rm cb$-norm  throughout the patch.

    Under the assumption $S_{\alpha}\neq \varnothing$, theorem~\ref{thm:ctg_topological} implies that the transitions between frames for $U_{\alpha}$ are locally constant. 
    We fix a point $\Phi_x\in U_{\alpha}$ and a curve $\Phi_{y(t)}\in M$ with $\Phi_{y(0)}=\Phi_x$ such that $\dot{\Phi}_{y(0)}=v\in T_xM$. On the patch $\alpha$ we have 
    \begin{equation}
        \dot{\Delta}_x^{\alpha}=\frac{d}{dt}\Delta^{\alpha}_{y(t)}\big\vert_{t=0} = d\Delta^{\alpha}_x(v)
    \end{equation}
    Within a local section, Lemma~\ref{lem:crossKL} implies $(\Delta_x^{\alpha})^*\Delta_{y(t)}^{\alpha}\geq 0$, so in particular
    \begin{equation}
        (\Delta_x^{\alpha})^*\Delta_{y(t)}^{\alpha}=(\Delta_{y(t)}^{\alpha})^*\Delta_x^{\alpha}.
    \end{equation}
    Differentiating this equation at $t=0$ yields
    \begin{equation}
        (\Delta_x^{\alpha})^*d\Delta^{\alpha}_x(v)=d\Delta^{\alpha}_x(v)^*\Delta_x^{\alpha}.
    \end{equation}
    We also have 
    \begin{equation}
        \frac{d}{dt}(\Delta^{\alpha}_{y(t)})^*\Delta^{\alpha}_{y(t)}\big\vert_{t=0}=d\Delta^{\alpha}_x(v)^*\Delta_x^{\alpha}+(\Delta_x^{\alpha})^*d\Delta^{\alpha}_x(v)=0,
    \end{equation}
    hence $(\Delta_x^{\alpha})^*d\Delta^{\alpha}_x(v)=\Lambda_x^{\alpha}d\Delta^{\alpha}_x (v)=0$. Since $x$ and $v$ were arbitrary, on the local section $\omega=\Lambda d\Delta=0$.

    Relative to the local section $(\Lambda_x^{\alpha}, \Delta_x^{\alpha})$, every other normalized frame over the same point is uniquely given by
    \begin{equation}
        \Delta = \Delta^{\alpha}_xg,\quad \Lambda=g^{-1}\Lambda_x^{\alpha},\qquad g=\Lambda_x^{\alpha}\Delta \in G= \Aut(A).
    \end{equation}
    Hence the section identifies the bundle over this patch with $\pi^{-1}(U_{\alpha})\equiv U_{\alpha}\times G$, and we recognize that $\pi:\overline{\Gr}_A(\CH)|_M\longrightarrow M$ is a principal bundle.
    In this local trivialization
    \begin{equation}
        \omega = g^{-1}dg,
    \end{equation}
    which is the flat product principal connection. Its horizontal lifts have constant $g$, so its local parallel transport is $\Delta_y^\alpha\Lambda_x^\alpha$, agreeing with Theorem~\ref{thm:ctg_topological}.
\end{proof}

Consequently, via theorem~\ref{thm:ctg_topological}, parallel transport is given locally by $\CT_{y\leftarrow x}^{\alpha}=\Delta_y^{\alpha}\Lambda_x^{\alpha}$ and depends only on the homotopy class of the path with fixed endpoints.
For a base frame $(\Phi_0, \Delta_0, \Lambda_0)$ the monodromy representation
\begin{equation}
 \rho:\pi_1(M,\Phi_0)\longrightarrow\operatorname{Aut}_*(A),
 \qquad
 \rho([\gamma])=\Lambda_0\mathcal T_\gamma\Delta_0
\end{equation}
determines logical automorphisms obtained by transport as 
$\operatorname{im}\rho\cong\pi_1(M,\Phi_0)/\ker\rho$.

This concludes our Knill-Laflamme-geometric survey of quantum fault tolerance. Throughout this section, we have seen how the generalized local joint KL conditions serve as sufficient conditions for constructing flat transport, and can be used to construct a frame bundle over a smooth manifold which realizes fault-tolerant logical operations as differential-geometric monodromy.

\section{Complex geometry of fault tolerance}\label{sec:complex_geom}
\subsection{Operator algebras in $\infty$-dimensions}

In the previous section we have discussed how a generalization of the Knill-Laflamme conditions together with the $\rm cb$-topology on the space of idempotents $\Gr_A(\CH)$ implies that the sub-Grassmannian bundle $\overline{\Gr_A(\CH)}\big\vert_M\xrightarrow{\pi} M\subset \Gr_A(\CH)$ carries a flat connection, and we have further shown that $M$ can be equipped with a smooth structure for appropriate choices of sets of jointly correctable idempotents.
In this section we consider smooth subspaces $M$ corresponding to quantum stabilizer codes. 
Quantum stabilizer codes are defined as SF-idempotents, where $\Delta:A\to B(\CH),\; \Delta(a)\Delta(b)=\Delta(ab)$ is a non-unital representation of the $C^*$ algebra $A$ such that 
\begin{equation}
    \Delta(a)s=s\Delta(a)=\Delta(a),\quad  \forall s\in \CS\; \forall a\in A,
\end{equation}
with $\CS\subset \CU(\CH)$ an abelian group. Conventionally, $\CS$ is taken to be a subgroup of the Heisenberg group $H(V)$, which acts irreducibly on $\CH$ and fits into the sequence 
\begin{equation}
    0\to U(1) \to H(V)\xrightarrow{\pi} V\to 0.
\end{equation}
Here $(V, \omega)$ is a symplectic vector space and operators $h_1, h_2 \in  H(V)$ satisfy
\begin{equation}
    h_1h_2=e^{-i2\pi \omega(v_1, v_2)}h_2h_1,\quad v_i=\pi(h_i).
\end{equation}

A common choice is a qudit system with\footnote{We denote $\Z_p=\Z/p\Z$.}  $V=\Z_p^{2n}$ for prime $p$, where
\begin{equation}
\omega_p(x,y)= x^TJy\; \mod p,\quad  J=\begin{pmatrix}
    0 & I_n \\ -I_n & 0
\end{pmatrix}
\end{equation}
determines the commutation phases $e^{i\frac{2\pi}{p}\omega_p(x,y) }$.

In this case stabilizer codes are described by isotropic subspaces of $V$. For finite $V$, the collection of isotropic subspaces is finite, and hence the corresponding subset $S\subseteq {\rm Gr}_A(\CH)$ is discrete (recall that ${\rm Gr}_A(\CH)$ carries the norm topology). 

In contrast, the Hilbert space $\CH=L^2(\R^n)$ carries a representation of the Heisenberg group $H(\R^{2n})$, where we take $\omega=x^TJy$ as the symplectic form over $V=\R^{2n}$. Up to a central character, displacement operators act as
\begin{equation}
    D\lr{\xi_1\oplus \xi_2}\psi(q) = e^{i\sqrt{2\pi}\xi_2^T(q-\sqrt{2\pi}\xi_1/2)}\psi(q-\sqrt{2\pi}\xi_1),
\end{equation}
and commute as $D(x)D(y)=e^{-i2\pi\omega(x,y)}D(y)D(x)$. Displacement operators are distributionally orthonormal, $\Tr[D(x)]=\delta(x),\, x\in \R^{2n}$, and for $f:\R^{2n}\to \C$ with $f\in L^1(\R^{2n})$, we define the integral\footnote{See ref.~\cite{Hall2013} for a broader discussion.}
\begin{equation}
    D[f]=\int_{\R^{2n}} f(x)D(x) dx.\label{eq:Bochner}
\end{equation}
For $f,g\in L^2(\R^{2n})$ we have $\Tr[D[f]^{\dagger}D[g]]=\langle f, g\rangle_{L^2(\R^{2n})}$, and as $\|D[f]\|\leq \|f\|_{L^1(\R^{2n})}$, one can also show that for the map $\CO_f(X)=D[f] X D[f]^{\dagger}$ we have the following norm inequalities.
\begin{lemma}
    Let $f\in L^1(\R^{2n})$,
    \begin{equation}
        D[f]=\int_{\R^{2n}} f(x)D(x) dx,
    \end{equation}
    and let $\CO[f](X)= D[f]XD[f]^{\dagger}$.
    It holds that
    \begin{enumerate}
        \item $\|\CO[f]\|_{\rm cb}= \|D[f]\|^2$,
        \item $\|\CO[f]-\CO[g]\|_{\rm cb} \leq (\|f\|_1+\|g\|_1)\| f-g\|_1$,
        \item $\|\CO[f]-\CO[g]\|_{\rm cb} \leq (\|f\|_2+\|g\|_2)\| f-g\|_2$ \hspace{1cm}for $f,g\in L^2(\R^{2n})$.
    \end{enumerate}
\end{lemma}
\begin{proof}
    Submultiplicativity of the operator norm gives $\|\CO[f]\|_{\rm cb}\leq \|D[f]\|^2$. \\
    Since $\|\CO[f](I)\|=\|D[f]\|^2$ attains the supremum, claim (1.) follows. For the next inequalities we use
    \begin{equation}
        AXA^{\dagger}-BXB^{\dagger} = (A-B)XB^{\dagger} +AX(A-B)^{\dagger},
    \end{equation}
    which implies claim (2.) using the triangle inequality, submultiplicativity, as well as $\|D[f]\|\leq \|f\|_{L^1(\R^{2n})}$.
    Claim (3.) follows similarly by using that  $\|D[f]\|\leq \|D[f]\|_{\rm HS}= \|f\|_{L^2(\R^{2n})}$.
\end{proof}

We restrict here to physical quantum states $\rho=D[f]$ such that $f\in \CS(\R^{2n})$ is a Schwartz function, first investigated in ref.~\cite{Keyl_2016}. This restriction ensures that arbitrary moments of physical position and momentum operators $\Tr[\rho\hat{q}^a\hat{p}^b]<\infty$ are well-defined. The relevant spaces are summarized via the Gelfand triple,
\begin{equation}
    \CS(\R^{2n})\subset L^2(\R^{2n})\subset \CS'(\R^{2n}).
\end{equation}
Here, elements of $\CS'(\R^{2n})$ represent generalized operators, while the subspace $O=\{f\in \CS'(\R^{2n}):\, f(x)=\overline{f(-x)}\}\subset \CS'(\R^{2n})$ represents physical observables.

\subsection{Segal-Bargmann representation of states and operators}
We now introduce the Segal-Bargmann representation of quantum mechanics. The previous discussion has emphasized that observables for $\CH=L^2(\R^n)$ are described by suitable functions over $\R^{2n}$. If such functions were the starting point for deriving a quantum-mechanical theory, as is investigated in geometric quantization~\cite{Schottenloher2025GeometricQuantization}, we would need to choose a subspace of $L^2(\R^{2n})$ unitarily equivalent to $L^2(\R^n)$. To this end, let
\begin{equation}
    \Pi = X(\Omega\; I_n) \in \C^{n\times 2n},\quad  X\in\GL_n(\C),\qquad \Omega\in \HH_n 
\end{equation}
be a matrix that identifies $\R^{2n}$ with $\C^n$, such that ${\rm rank}(\Pi^{\dagger}\Pi)=n$, which we have already factorized into a general orientation $X$ and the polarization $\Omega$. 
\begin{equation}
    \HH_n=\{Y=Y^T\in M_{n}(\C):\, \Im(Y)>0\}
\end{equation}
denotes the Siegel upper half-space. Up to the general change of coordinates $z\mapsto Xz$, the unitary Bargmann transform
\begin{align}
    (\CB_{\Omega}\psi)(z)&= c_{\Omega}e^{\frac{\pi}{2}z^T(\Im\Omega)^{-1}z}\int_{\R^n}e^{\frac{i}{2}q^T\Omega q +i\sqrt{2\pi}z^Tq}\psi(q)dq,\nonumber \\
    c_{\Omega}&=\pi^{-n/4}\det(\Im\Omega)^{1/4}.
\end{align}
identifies functions $\psi\in L^2(\R^n)$ with holomorphic functions $F\in L^2(\C^n, \mu_{\Omega})$ normalizable under the inner product
\begin{equation}
    \langle F, G \rangle_{\Omega} = \int_{\C^n }\overline{F(z)}G(z) d\mu_{\Omega}(z),\quad d\mu_{\Omega}(z)=\frac{e^{-\pi H_{\Omega}(z, z)}}{|\det(\Im\Omega) |}dz,
\end{equation}
with Hermitian form $H_{\Omega}(x,y)=x^{\dagger}(\Im\Omega)^{-1}y$.
The overall change of frame $z\mapsto Xz$ then corresponds to the choice $ (\CB_{\Pi}\psi)(z)=(\CB_{\Omega}\psi)(X^{-1}z)$ and replaces $\Im\Omega$ with $Q=X(\Im \Omega)X^{\dagger}$ in $\mu_{\Omega}$.

Under the above map, displacement operators $D(z)$ with $z=(z_1, z_2)\in \R^{2n}$ become 
\begin{align}
    T(v)&=\CB_{\Omega}D(z)\CB_{\Omega}^{-1},\quad v=\Omega z_1+z_2, \nonumber\\ 
    T(v)F(z)&=e^{-\frac{\pi}{2}H_{\Omega}(v,v)-\pi H_{\Omega}(v,z)}F(z+v).
\end{align}
The Bargmann transform $\CB_{X, \Omega}$ preserves the commutation relations of displacements. Using the expression derived above, we can directly verify that
\begin{equation}
    T(x)T(y)=e^{i2\pi \Im H_{\Omega}(x,y)}T(y)T(x).
\end{equation}

Symplectic changes of frame on $\R^{2n}$ via $S=\lr{\begin{smallmatrix}
        a & b \\ c & d
    \end{smallmatrix}}\in \Sp_{2n}(\R)$ preserve and act transitively on $\HH_n$ with
\begin{equation}
    (c\Omega+d)^{-T}(\Omega\; I_n)S^T=(S.\Omega\, I_n),\quad
    S.\Omega=(a\Omega+b)(c\Omega+d)^{-1}, \label{eq:modular_trans}
\end{equation}
together with the normalized coordinate transformation
$S.z=(c\Omega+d)^{-T}z$. According to this action of $\Sp_{2n}(\R)$,
$\HH_n$ labels choices for symplectically equivalent representations
of the Heisenberg group via $T(z)$.
We use the convention $\Pi\mapsto\Pi S^T$, corresponding to the real
column transformation $v\mapsto S^Tv$. Equation~\eqref{eq:modular_trans}
then gives the Segal-Bargmann frame change
\begin{equation}
 (\Omega,X)\longmapsto\bigl(S.\Omega,\,X(c\Omega+d)^T\bigr).
\end{equation}
For example, the Segal-Bargmann vacuum $F(z)=1$ has preimage
$\CB_{\Omega}^{-1}1=c_{\Omega}e^{-\frac{i}{2}q^T\overline{\Omega}q}
=\phi_{\Omega}(q)$. A metaplectic operator $U_{S^T}$ implementing
$S^T$, with $U_{S^T}D(v)U_{S^T}^{\dagger}=D(S^Tv)$, acts as
\begin{equation}
 U_{S^T}\phi_{\Omega}
 =e^{i\chi(S^T,\Omega)}\phi_{S^{-1}.\Omega},
\end{equation}
where $\chi(S^T,\Omega)$ is the metaplectic phase
(see, e.g., ref.~\cite{Hall2013}).

\subsection{Tempered distribution spaces}

The action of operators $A$ on holomorphic functions can be written as 
\begin{equation}
    \int_{\C^n} K_A(z, \overline{w}) F(w) \mu_{\Omega}(w)
\end{equation}
for $K_A(z, \overline{w})$ holomorphic in the first and antiholomorphic in the second variable. For trace-class $A$, the trace is expressed as
\begin{equation}
    \Tr[A]=\int_{\C^n}K_A(z,\overline{z}) d\mu_{\Omega}(z)
\end{equation}
and membership in the Schwartz or dual Schwartz operator space is summarized by the following conditions~\cite{Toft2016}
\begin{align}
    C_A(z, \overline{w})&= e^{-\frac{\pi}{2}[H(z, z)+H(w,w)]}K_A(z, \overline{w}),\nonumber\\
    A\in \CS(\R^{2n}) &\Leftrightarrow |C_A(z, \overline{w})|=O((1+\|z\|+\|w\|)^{-N})\;\; \forall N\geq 1,\nonumber\\
    A\in \CS'(\R^{2n}) &\Leftrightarrow \exists N:\;\; |C_A(z, \overline{w})|=O((1+\|z\|+\|w\|)^{N}).\label{eq:dual_schwartz}
\end{align}

In the following we define a Hilbert space containing the Segal-Bargmann image of $\CS'(\R^{2n})$. This is done by defining an inner product relative to a Schwartz kernel, or, similarly, by considering a regularization $X_{\beta}=R_{\beta}XR_{\beta}$, $R_{\beta}=(1+\beta)^{-\hat{N}/2}$ and the normalized inner product $\langle X, Y\rangle_{\beta} = N_{\beta}\Tr[X_{\beta}^{\dagger}Y_{\beta}]$. 
To define the kernel space, we denote by $\CO(\C^n)$ the space of holomorphic functions $f: \C^n\to \C$. The following construction can be understood as an adaptation of the Gelfand-Naimark-Segal construction~\cite{GelfandNaimark1943,Segal1947}.

\begin{definition}\label{def:K_norm}
      Let $K\in \CO(\C^{2n})$ and let $H$ be a positive definite Hermitian form on $\C^n$. The $\beta>0$ regularized inner product is 
      \begin{align}
    \langle K,L\rangle_{H,\beta}
    &=\left(\frac{1+\beta}{\pi}\right)^{2n}
    \int_{\C^n\times\C^n}
    \overline{K(z,\overline w)}L(z,\overline w) \nonumber\\
    &\qquad\times e^{-\pi(1+\beta)[H(z,z)+H(w,w)]}\,dz\,dw,
      \end{align}
      and induces a norm $\|K\|_{H, \beta}=\sqrt{\langle K, K\rangle_{H,\beta}}$.
The associated Hilbert space is the space of entire functions
\begin{equation}
 \mathscr K_{H,\beta}
 =\{K\in \CO(\C^{2n}):\ \|K\|_{H,\beta}<\infty\},
\end{equation}
which we associate with the set of Bargmann kernels $K(v,\overline{w})$.
\end{definition}
We can verify that this definition indeed produces a Hilbert space, in particular that it is complete.
\begin{lemma}\label{lem:kernel_space_complete}
The space $\mathscr K_{H,\beta}$ is complete.
\end{lemma}
\begin{proof}
Substituting $\eta=\overline w$ and writing $x=(z,\eta)$ gives
\begin{align}
\|K\|_{H,\beta}^2&=\int_{\C^{2n}}|K(x)|^2\rho(x)\,dx, \nonumber\\
\rho(z,\eta)&=\left(\frac{1+\beta}{\pi}\right)^{2n}
e^{-\pi(1+\beta)[H(z,z)+H(\overline\eta,\overline\eta)]}.
\end{align}
Let $B(x,1)$ be the Euclidean unit ball centered at $x$ and $v$ its volume.
Since the real and imaginary parts of $K$ are harmonic, the mean-value
property followed by Cauchy--Schwarz gives
\begin{equation}
|K(x)|^2
=\left|\frac1v\int_{B(x,1)}K(u)\,du\right|^2
\leq\frac1v\int_{B(x,1)}|K(u)|^2\,du.
\end{equation}
For a compact set $E\subset\C^{2n}$, the set
$E_1=\{u:\operatorname{dist}(u,E)\leq1\}$ is compact, so the continuous,
strictly positive weight satisfies $m_E:=\min_{E_1}\rho>0$.
For every $x\in E$, we have $B(x,1)\subset E_1$, and therefore
\begin{equation}
|K(x)|^2
\leq\frac1{v m_E}\int_{B(x,1)}|K(u)|^2\rho(u)\,du
\leq\frac1{v m_E}\|K\|_{H,\beta}^2.
\end{equation}
Taking the supremum over $E$ and applying this estimate to differences gives
\begin{equation}
\sup_E|K_j-K_\ell|
\leq(v m_E)^{-1/2}\|K_j-K_\ell\|_{H,\beta}.
\end{equation}
Thus a Cauchy sequence $(K_j)$ for our norm is uniformly Cauchy on every compact set and converges locally uniformly to an entire function $K$.

It remains to prove convergence in norm.
Given $\varepsilon>0$, choose $N$ such that
$\|K_j-K_\ell\|_{H,\beta}\leq\varepsilon$ for $j,\ell\geq N$.
For each fixed $j\geq N$, Fatou's lemma applied to the nonnegative
functions $|K_j-K_\ell|^2\rho$ yields
\begin{align}
\int_{\C^{2n}}|K_j-K|^2\rho\,dV
&=\int_{\C^{2n}}\liminf_{\ell\to\infty}|K_j-K_\ell|^2\rho\,dV \nonumber \\
&\leq\liminf_{\ell\to\infty}\|K_j-K_\ell\|_{H,\beta}^2
\leq\varepsilon^2.
\end{align}
In particular, $K_j-K$ has finite weighted norm. Since $K_j$ also has
finite norm, the triangle inequality gives $K\in\mathscr K_{H,\beta}$.
The same estimate proves $\|K_j-K\|_{H,\beta}\to0$.
\end{proof}

We may now also define the space of finite-dimensional subspaces of this space, analogous to the code idempotents we considered earlier. For simplicity, here we restrict to SF-idempotents.
\begin{definition}\label{def:GKPGR}
    For $K\geq1$, let $\operatorname{Gr}_K(\mathscr K_{H,\beta})$ be the metric space of complex $K$-dimensional subspaces of $\mathscr K_{H,\beta}$. Each such subspace $V$ is closed and has a bounded orthogonal projection $\Phi_V$, and the projection distance on $\operatorname{Gr}_K(\mathscr K_{H,\beta})$ is given by
    \begin{equation}\label{eq:gaussian_grassmann_gap}
    d_{H,\beta}(V,W)
    =\|\Phi_V-\Phi_W\|_{B(\mathscr K_{H,\beta})}
    \eqqcolon \sup_{\|X\|_{H,\beta}=1}
    \|(\Phi_V-\Phi_W)X\|_{H,\beta}.
\end{equation}
\end{definition}

For the following discussion, the space $\operatorname{Gr}_K(\mathscr K_{H,\beta})$ takes the place of the Grassmannian of code-idempotents $\Gr_A(\CH)$ considered in section~\ref{sec:KL_Geometry}. 

\subsection{Gottesman-Kitaev-Preskill codes}
We are finally ready to define Gottesman-Kitaev-Preskill (GKP) codes~\cite{Conrad_2022, GKP}. These are SF quantum stabilizer codes for $A=\C[\CG]$, where $\CG=\CU H(\Z^n/D\Z^n)$ is the standard unitary representation of the corresponding Heisenberg group $H(\Z^n/D\Z^n)$ and $D={\rm diag}(d_1,\hdots, d_n),\; d_i|d_{i+1}$ the type of the code.

For the following discussion we fix the Bargmann polarization
$\Omega_B=iI_n$, with $\Pi_B=(iI_n\;I_n)$ and Hermitian form
$H_B(z_B,w_B)=z_B^{\dagger}w_B$.
A real GKP lattice $S^TL_D$, where $L_D=\Z^n\oplus D\Z^n$ and
$S=\left(\begin{smallmatrix}a&b\\c&d\end{smallmatrix}\right)
\in\Sp_{2n}(\R)$, has complex image
\begin{align}
 \Lambda_B&=\Pi_BS^TL_D\nonumber
 =X(\Omega\Z^n+D\Z^n),\quad  X=(ic+d)^T\\
 \Omega&=X^{-1}(ia^T+b^T)=S.\Omega_B\in\HH_n.
\end{align}
Thus $\Pi_BS^T=X(\Omega\;I_n)$, as in eq.~\eqref{eq:modular_trans}, describes
the same lattice using polarization $\Omega$ on the reference real
phase space, followed by the change of coordinates $X$.
Since $S$ is symplectic, we have that
\begin{equation}
     X\Im(\Omega)X^{\dagger}=I_n,\qquad
 H_B(Xu,Xv)=H_{\Omega}(u,v).
\end{equation}
In coordinates $u=X^{-1}z_B$ the lattice is
$\Omega\Z^n+D\Z^n$ and the ambient Hermitian form becomes $H_{\Omega}$.
We use these normalized coordinates for the subsequent discussion and
return to $z_B$ when comparing different codes.

The GKP code is defined via stabilizers 
\begin{equation}
    S(\lambda)=\chi_s(\lambda)T(\lambda),\quad  \lambda\in \Lambda=\Omega\Z^n\oplus D\Z^n,
\end{equation}
with $\chi(\lambda)$ a semicharacter. For the stabilizer group to be abelian, it is required that $\Im H_{\Omega}(\Lambda,\Lambda)\in \Z$, i.e., that the torus $X_{\Pi}=\C^n/\Lambda$ has the structure of a polarized abelian variety (see ref.~\cite{mayrand2026complex} for a detailed account of this perspective). 

Let $\Lambda^{\perp}$ be the dual of $\Lambda$ under $\Im H_{\Omega}$. The torus $X_{\Omega D} = \C^n/(\Omega \Z^n + D\Z^n)$ is a $(\det D)^2$-fold cover of $X^{\perp}_{\Omega D} = \C^n/(\Omega D^{-1}\Z^n + \Z^n)$, and $s\in X^{\perp}_{\Omega D}$ labels a syndrome sector.

Using this new toolbox, we can determine the form of $\Delta$ for these codes. Let $a\in A$. %
The conditions $S(\lambda)\Delta(a)=\Delta(a)S(\lambda)=\Delta(a)$ imply %
that 
\begin{equation}
    K_{\Delta(a)}(z, \overline{w}) = G_{\Omega, s}(z)\overline{G_{\Omega, s}(w)} \sum_{\alpha, \beta\in D^{-1}\Z^n / \Z^n} M_{\alpha\beta}(a)\vartheta_{\alpha}(z-s, \Omega)\overline{\vartheta_{\beta}(w-s,\Omega)},\label{eq:KDelta}
\end{equation}
with $G_{\Omega, s}(z)=e^{\frac{\pi}{2}z^T(\Im \Omega)^{-1}z -i2\pi (z-s/2)^T(\Im\Omega)^{-1}\Im s}$, and $\vartheta_{\alpha}(z, \Omega)=\vartheta_{\alpha,0}(z,\Omega)$, where 
\begin{equation}
    \vartheta_{\alpha,\beta}(z, \Omega)=\sum_{m \in \Z^n} e^{i\pi (m+\alpha)^T\Omega (m+\alpha) +i2\pi (m+\alpha)^T(z+\beta) }
\end{equation}
are the classical theta functions. These identify a holomorphic line bundle $L_s$ on the torus $X_{\Omega, D}=\C^n/(\Omega\Z^n + D\Z^n)$ with $\dim H^0(X_{\Omega, D}, L)=\det D$. $H^0(X_{\Omega, D}, L_s)$ has an inner product
\begin{equation}
    \langle f, g \rangle_{\Omega, s}=\int_{X_{\Omega, D}} \overline{f(z)}g(z) e^{-2\pi (\Im z - \Im s)^T (\Im \Omega)^{-1}(\Im z- \Im s)} dz
\end{equation}
and provides representations $\overline{T}(c)$ for the finite Heisenberg group $H(\Z^n/D\Z^n)$,
\begin{equation}
    \overline{T}(\Omega \alpha +\beta)f(z)=e^{i\pi \alpha^T\Omega\alpha+i2\pi\alpha^T(z-s+\beta)}f(z+\Omega\alpha+\beta)
\end{equation}
given by the image of $c=(\alpha, \beta)\in (D^{-1}\Z^n\oplus \Z^n)/(\Z^n\oplus D\Z^n)$.

$K_{\Delta(a)}$ satisfies the tempered distribution growth condition of eq.~\eqref{eq:dual_schwartz}.
\begin{lemma}\label{lem:finite_K_GKP}
    For the GKP code with $\Omega\in \HH_n$ and frame map $\Delta$, we have $\forall a\in A$,
    \begin{equation}
        |e^{-\frac{\pi}{2}[H_{\Omega}(z, z)+H_{\Omega}(w,w)]}K_{\Delta(a)}(z, \overline{w})|\leq \frac{\|M(a)\|_{\ell^1}}{\det(\Im\Omega)}(1+\|\Im\Omega\|)^{2n}.\label{eq:bound_KGKP}
    \end{equation}
    In particular, for $\beta>0$, $\|K_{\Delta(a)}\|_{H_{\Omega}, \beta}<\infty$.
\end{lemma}
\begin{proof}
    Using $\Omega = X+iY$, $z=x+iy$, we have \begin{equation}|e^{-\frac{\pi}{2}H_{\Omega}(z,z)}G_{\Omega}(z)\vartheta_{\alpha}(z, \Omega)| = e^{-\pi y^TY^{-1}y}|\vartheta_{\alpha}(z, \Omega)|.\end{equation}
Following the derivation of ref.~\cite[eq. (3)]{Deconinck2004},  Poisson summation yields
\begin{align}
    e^{-\pi y^TY^{-1}y}|\vartheta_{\alpha}(z, \Omega)|&\leq \det(Y)^{-\frac{1}{2}}\sum_{k\in \Z^n} e^{-\pi k^TY^{-1}k} \nonumber\\ 
    &\leq \det(Y)^{-\frac{1}{2}}\sum_{k\in \Z^n} e^{-\pi k^Tk/\|Y\|}\nonumber \\
    &\leq \det(Y)^{-\frac{1}{2}} (1+\|Y\|)^n,
\end{align}
where the last inequality follows from the geometric series, and $\|K_{\Delta}(a)\|_{\beta}<\infty$ follows from def.~\ref{def:K_norm}.
\end{proof}

Writing the kernel in eq.~\eqref{eq:KDelta} as $K_{\Omega,s;a}$, the change to physical coordinates gives
\begin{equation}
 K^B_{\Omega,X,s;a}(z_B,\overline{w_B})
 =K_{\Omega,s;a}\bigl(X^{-1}z_B,\overline{X^{-1}w_B}\bigr).
\end{equation}
With $H_B(Xz,Xz)=H_{\Omega}(z,z)$ and the Jacobian of
$(z,w)\mapsto(Xz,Xw)$ in def.~\ref{def:K_norm},  one has
 $\|K^B_{\Omega,X,s;a}\|_{H_B,\beta}
 =|\det X|^2\|K_{\Omega,s;a}\|_{H_{\Omega},\beta}$,
so Lemma~\ref{lem:finite_K_GKP} also gives
$K^B_{\Omega,X,s;a}\in\mathscr K_{H_B,\beta}$.

In the following two lemmas we show that, for fixed $D$, the spaces of logical GKP operators $V_p$ are continuous in the parameters $p\in \HH_n\times\C^n\times U(n)$ and that the map $p\mapsto V_p$ is indeed locally a smooth embedding that descends to a global embedding on the appropriate quotient space. For clarity, in each proof we first fix the rotation and then allow it to vary.

Let $D$ be fixed and $d=\det D$. We set $\mathcal U_0\subset\HH_n\times\C^n\times U(n)$ to be an open
neighbourhood whose rotation component lies in one coordinate chart $R:W\subset\R^{n^2}\to U(n)$. $R(\eta)$ can be chosen to have the form $e^{i\sum_{i=1}^{n^2}\eta_i {\mathfrak{r}_i}}$ with $\mathfrak{r}_i$ a basis of Hermitian matrices.

Write $p=(\Omega,s,R(\eta))$ and choose a smooth frame $X_0(\Omega,s)$ satisfying
$X_0\Im\Omega X_0^{\dagger}=I_n$, such as
$X_0=(\Im\Omega)^{-1/2}$. Setting
$X(p)=R(\eta)X_0(\Omega,s)$ gives
\begin{equation}\label{eq:gkp_physical_frame}
 X(p)\Im\Omega X(p)^{\dagger}=I_n.
\end{equation}
The physical syndrome is given by the local lift
$s_B=X(p)s$, and, for $\alpha\in D^{-1}\Z^n/\Z^n$, we set
\begin{align}
 f_{\alpha,p}(u)&=G_{\Omega,s}(u)\vartheta_\alpha(u-s,\Omega), \quad &&f^B_{\alpha,p}(z)=f_{\alpha,p}(X(p)^{-1}z)\nonumber\\
 K_{p,\alpha\gamma}(u,\overline v)
 &=f_{\alpha,p}(u)\overline{f_{\gamma,p}(v)},\quad &&K^B_{p,\alpha\gamma}(z,\overline w)
 =f^B_{\alpha,p}(z)\overline{f^B_{\gamma,p}(w)}.
 \label{eq:gkp_physical_kernels}
\end{align}
Here $z,w$ are physical coordinates. For a logical operator with
coefficient matrix $M(a)$ in this basis, we write
\begin{equation}\label{eq:gkp_physical_kernel_span}
 K^B_{p,a}=\sum_{\alpha,\gamma}M_{\alpha\gamma}(a)K^B_{p,\alpha\gamma},\quad V_p=\operatorname{span}_{\C}\{K^B_{p,\alpha\gamma}\}
 \subset\mathscr K_{H_B,\beta}.
\end{equation}
The theta functions have distinct cosets of Fourier frequencies, so the
$f_{\alpha,p}$ and their product kernels are linearly independent.
Since $X(p)$ is invertible, we have $\dim V_p=d^2$.

The parameters specifying a GKP code with fixed $D$ are $p=(\Omega,s,\eta)\in \HH_n\times\C^n\times \R^{n^2}$, for which we use the norm
\begin{equation}\label{eq:gkp_parameter_norm}
 \|p\|=\|\Omega\|_F+\|s\|_2+\|\eta\|_2.
\end{equation}

\begin{lemma}[Local Lipschitz continuity of GKP kernels]\label{thm:gkp_kernel_lipschitz}
Fix $p_0\in\mathcal U_0$ and $\beta>0$.
Assume that $M(a)$ in eq.~\eqref{eq:gkp_physical_kernel_span}
is independent of $p$.
There exist a neighbourhood $U\subset\mathcal U_0$ of $p_0$
and a constant $C_{U,\beta}>0$ such that, for every
$p,p'\in U$ and $a\in A$,
\[
 \|K^B_{p,a}-K^B_{p',a}\|_{H_B,\beta}
 \le
 C_{U,\beta}\|M(a)\|_{\ell^1}\|p-p'\|
\]
where $\|M(a)\|_{\ell^1}=\sum_{\alpha,\gamma}
|M_{\alpha\gamma}(a)|$.
The kernels and the orthogonal projections onto $V_p$ depend smoothly
on $p$. In particular, there is a constant $\widetilde C_{U,\beta}>0$ such that
\begin{equation}
 d_{H_B,\beta}(V_p,V_{p'})\leq\widetilde C_{U,\beta}\|p-p'\|.\label{eq:V_upperbound}
\end{equation}
\end{lemma}

\begin{proof}
A fixed rotation acts unitarily on $\mathscr K_{H_B,\beta}$ by rotating
both kernel arguments, preserving the kernel norm and the projection distance.
We may therefore first write $X=X_0$ and fix the rotation. Choose $U$ with compact closure in $\mathcal U_0$ and convex coordinate image, and constants $c,R_s>0$
such that, throughout $U$,
\begin{equation}
     Y:=\Im\Omega\ge cI_n,\qquad
 \|s\|_2\le R_s.
\end{equation}

The frame $X$, its inverse and their derivatives are bounded on this
closure.
For the segment $p_t=p+t(p'-p)$, Minkowski's inequality gives
\begin{equation}
 \|K^B_{p',a}-K^B_{p,a}\|_{H_B,\beta}
 \leq
 \int_0^1\|\dot K^B_{p_t,a}\|_{H_B,\beta}\,dt.
 \label{eq:Minkowski}
\end{equation}
It therefore suffices to bound $\|\dot K^B_{p_t,a}\|_{H_B,\beta}$.

All constants below may depend on this neighbourhood and the fixed characteristic set, but not on its points, the characteristic, or $a$. For a characteristic $\alpha$, we set
\begin{equation}
 f_{\alpha,p}(z)
 =G_{\Omega,s}(z)\vartheta_\alpha(z-s,\Omega),
 \qquad
 E_\Omega(z)=
 \exp\!\left(\frac{\pi}{2}z^\dagger Y^{-1}z\right).
\end{equation}
We first obtain uniform bounds for these functions and their real directional derivatives.
Let
\begin{equation}
 A=Y^{-1},\qquad r_m=m+\alpha,\qquad
 u=A\Im(z-s)\in\R^n.
\end{equation}
Completing the square gives the identity
\begin{equation}
 \left|
 G_{\Omega,s}(z)
 e^{i\pi r_m^T\Omega r_m+i2\pi r_m^T(z-s)}
 \right|
 =
 E_\Omega(z)
 e^{-\pi(r_m+u)^TY(r_m+u)}
\end{equation}
and for each nonnegative integer $j$, we define
\begin{equation}
 S_j(u)=
 \sum_{m\in\Z^n}\|r_m\|_2^j
 e^{-\pi(r_m+u)^TY(r_m+u)}.
\end{equation}

For each $0\leq j < \infty$, the function
\begin{equation}
 v\longmapsto
 \sum_{m\in\Z^n}\|m+v\|_2^j
 e^{-\pi c\|m+v\|_2^2}
\end{equation}
is uniformly bounded, and one can show that its defining series converges. 
Using $Y\geq cI_n$ and
\begin{equation}
 \|r_m\|_2^j
 \le 2^{j-1}(\|r_m+u\|_2^j+\|u\|_2^j),
\end{equation}
together with a bound on $x^je^{-\lambda x}\leq C_{j,\lambda}\,\forall x>0$,
one can show that, for every $j\geq 0$, one has
\begin{equation}
 S_j(u)\le C_j(1+\|u\|_2^j).\label{eq:Sbounds}
\end{equation}

By assumption, we have $\|u\|_2\le c^{-1}(\|z\|_2+R_s)$, so that we obtain the bound
\begin{equation}
 |f_{\alpha,p}(z)|\le C E_\Omega(z).
\end{equation}

Let $H=H^T$ be a direction in the $\Omega$ variable
and write $B=\Im H$, $\tau=\Im s$.
Since the derivative of $A$ in this direction is
$-ABA$, we have
\begin{equation}
 \ell_\Omega(z;H)
 :=\partial_\Omega\log G_{\Omega,s}(z)[H]
 =
 -\frac{\pi}{2}z^TABAz
 +i2\pi(z-s/2)^TABA\tau.
\end{equation}
Thus
\begin{equation}
 |\ell_\Omega(z;H)|
 \le C\|H\|_F(1+\|z\|_2^2).
\end{equation}
We also need to bound the derivatives
\begin{align*}
 \partial_\Omega f_{\alpha,p}(z)[H]
 =
 G_{\Omega,s}(z)
 \sum_{m\in\Z^n}
 \bigl(\ell_\Omega(z;H)+i\pi r_m^THr_m\bigr)
 e^{i\pi r_m^T\Omega r_m+i2\pi r_m^T(z-s)},
\end{align*}
and using the preceding bounds we obtain
\begin{equation}
 |\partial_\Omega f_{\alpha,p}(z)[H]|
 \le
 C\|H\|_F(1+\|z\|_2^2)E_\Omega(z).
\end{equation}

Similarly, for a direction $\sigma\in\C^n$ in the
$s$ variable, put $\nu=\Im\sigma$. Then
\begin{equation}
 \ell_s(z;\sigma)
 :=\partial_s\log G_{\Omega,s}(z)[\sigma]
 =
 i\pi\sigma^TA\tau
 -i2\pi(z-s/2)^TA\nu,
\end{equation}
so that
\begin{equation}
 |\ell_s(z;\sigma)|
 \le C\|\sigma\|_2(1+\|z\|_2).
\end{equation}
Since differentiating the theta series contributes
$-i2\pi r_m^T\sigma$, we obtain
\begin{equation}
 |\partial_s f_{\alpha,p}(z)[\sigma]|
 \le
 C\|\sigma\|_2(1+\|z\|_2)E_\Omega(z).
\end{equation}
For the derivative in $z$ (which comes from varying the frame $X(p)$), we have
\begin{equation}
 \partial_z\log G_{\Omega,s}(z)[v]
 =\pi v^TAz-2\pi i\,v^TA\tau,
\end{equation}
where $v\in\C^n$. Differentiating a theta summand contributes $2\pi i\,r_m^Tv$.
The bounds on $S_0$ and $S_1$ therefore give
\begin{equation}\label{eq:gkp_spatial_derivative_bound}
 |\partial_zf_{\alpha,p}(z)[v]|
 \le C\|v\|_2(1+\|z\|_2)E_\Omega(z).
\end{equation}

We set $C_p=X(p)^{-1}$. For a real parameter direction $h=(H,\sigma,0)$,
the chain rule gives
\begin{equation}
 df^B_{\alpha,p}(z)[h]
 =\partial_pf_{\alpha,p}(C_pz)[h]
  +\partial_zf_{\alpha,p}(C_pz)[(dC_p[h])z].
\end{equation}
By eq.~\eqref{eq:gkp_physical_frame},
\[
 E_\Omega(C_pz)=\exp\!\left(\frac\pi2\|z\|_2^2\right)
 =:E_B(z).
\]
Using the local bounds on $C_p$ and $dC_p$, we obtain
\begin{align}
 |f^B_{\alpha,p}(z)|&\le CE_B(z),\quad
 |\mathrm df^B_{\alpha,p}(z)[h]|
 \le C\|h\|(1+\|z\|_2^2)E_B(z).
 \label{eq:gkp_physical_derivative_bound}
\end{align}
Since $M(a)$ is constant, the product rule yields
\begin{align*}
 |\mathrm dK^B_{p,a}(z,\overline w)[h]|
 \le C\|M(a)\|_{\ell^1}\|h\|(1+\|z\|_2^2+\|w\|_2^2)E_B(z)E_B(w).
\end{align*}
Multiplying the squared bound by the weight in def.~\ref{def:K_norm} gives
\begin{align*}
 &|\mathrm dK^B_{p,a}(z,\overline w)[h]|^2
 e^{-\pi(1+\beta)(\|z\|_2^2+\|w\|_2^2)}\\
 &\quad\le C\|M(a)\|_{\ell^1}^2\|h\|^2
 (1+\|z\|_2^2+\|w\|_2^2)^2
 e^{-\pi\beta(\|z\|_2^2+\|w\|_2^2)}.
\end{align*}
The right-hand side is integrable for every $\beta>0$, uniformly on
$U$. Hence
\begin{equation}
 \|\mathrm dK^B_{p,a}[h]\|_{H_B,\beta}
 \le C_{U,\beta}\|M(a)\|_{\ell^1}\|h\|,
\end{equation}
and eq.~\eqref{eq:Minkowski} proves the first bound: the kernels $K^B_{p,a}$ are Lipschitz continuous in $p$. 

Following the same reasoning, using the bound in eq.~\eqref{eq:Sbounds}, it can be shown that, for any multiindex $\nu$, we have
\begin{equation}
    |\partial_p^{\nu} K_{p,\alpha\gamma}^B(z, \overline{w})|\leq C_{\nu}(1+\|z\|+\|w\|)^{2|\nu|} e^{\frac{\pi}{2}(\|z\|^2+\|w\|^2)}.
\end{equation}
This then implies that $K_{p,\alpha\gamma}^B(z, \overline{w})$ are smooth functions of $p$.

To obtain the corresponding bound for the subspaces $V_p$, we consider their orthogonal projections.
Let $T_p:\C^{d^2}\to\mathscr H$ have these kernels as columns. These operators play a role similar to that of $\CE_p\Delta$ in sec.~\ref{sec:KL_Geometry}.
Their linear independence makes $T_p^*T_p$ positive definite, so
\begin{equation}\label{eq:gkp_physical_projection}
 \Phi(p)=T_p(T_p^*T_p)^{-1}T_p^*
\end{equation}
is smooth in operator norm and idempotent, with $ \Phi(p)^2= \Phi(p)$. Its derivative is
uniformly bounded on $\overline{U}$. Integration along the segment from $p$ to $p'$ then proves the upper bound in eq.~\eqref{eq:V_upperbound}.
Finally, choose $X=R(\eta)X_0(\Omega,s)$ and allow $\eta$ to vary. Then $C_p=X^{-1}=X_0(\Omega,s)^{-1}R(\eta)^{\dagger}$.
The normalized functions $f_{\alpha,p}$ are independent of $\eta$, so
\begin{equation}
 \partial_\eta f^B_{\alpha,p}(z)[\xi]
 =\partial_z f_{\alpha,p}(C_pz)[(\partial_\eta C_p[\xi])z].
\end{equation}
On $\overline U$, we have $\|\partial_\eta C_p[\xi]\|\le C\|\xi\|_2$.
Since $X\Im\Omega X^{\dagger}=I_n$,
eq.~\eqref{eq:gkp_spatial_derivative_bound} gives the same Gaussian
derivative bound for every direction $h=(H,\sigma,\xi)$.
The same estimates for higher and mixed derivatives imply smoothness, and integration on a smaller convex neighbourhood gives the kernel bound, so
eq.~\eqref{eq:gkp_physical_projection} applied to the rotated kernel
columns proves the main statement.
\end{proof}

This lemma shows that the kernels $K_{p,a}^B$ are Lipschitz continuous and, furthermore, smooth functions of the parameter $p\in U$. This then also applies to the distribution spaces $V_p$. In the following, we further show that these subspaces can be continuously identified with their parameters $p$.

\begin{lemma}[Local lower bound for GKP kernel subspaces]\label{lem:gkp_kernel_lower}
Let $V_p$ be as in eq.~\eqref{eq:gkp_physical_kernel_span},
and fix $p_0\in\mathcal U_0$ and $\beta>0$.
There are a neighbourhood $U\subset\mathcal U_0$ of $p_0$ with convex
coordinate image and a constant $c_{U,\beta}>0$
such that, for all
$p,p'\in U$,
\begin{equation}\label{eq:gkp_grassmann_lower}
 d_{H_B,\beta}(V_p,V_{p'})\geq c_{U,\beta} \|p-p'\|.
\end{equation}
\end{lemma}
\begin{proof}
As above, fixed rotations preserve the distance, so we first write $X=X_0$ and restrict all parameter differences and differentials to a fixed rotation. We retain $\dot X$ in the kernel calculation to handle joint variations at the end. By Lemma~\ref{thm:gkp_kernel_lipschitz}, the orthogonal projection $\Phi(p)$ in eq.~\eqref{eq:gkp_physical_projection} is smooth in the
fixed space $\mathscr H=\mathscr K_{H_B,\beta}$.

The key point of this statement is to show that $\mathrm d\Phi(p)$ is injective. To this end, let
$h=(H,\sigma,0)$, $H=H^T$, be a real tangent direction, and denote differentiation along $p+th$ at $t=0$ by a dot, i.e., $\dot{\Phi}=d\Phi(p)[h]$. Our goal is to show that $\dot \Phi=0$ must also imply $h=0$.

If $\dot \Phi=0$, differentiating
$\Phi(p)K^B_{p,00}=K^B_{p,00}$ yields
\begin{equation}
 (I-\Phi(p))\dot K^B_{p,00}=\dot \Phi\,K^B_{p,00}=0,
\end{equation}
so that $\dot K^B_{p,00}\in {\rm span}\{K^B_{p,\alpha\gamma}\}$. Using $X=X(p)$ in the
coordinate substitution, we define
\begin{equation}
 L(u,\overline v):=\dot K^B_{p,00}(Xu,\overline{Xv}),
 \qquad A_X=X^{-1}\dot X,
\end{equation}
where we have $L\in\operatorname{span}\{K_{p,\alpha\gamma}\}$.
Differentiating at fixed physical argument $z_0=Xu$ gives
\begin{equation}
 \dot f^B_{\alpha,p}(Xu)
 =\dot f_{\alpha,p}(u)-(A_Xu)^T\nabla_u f_{\alpha,p}(u),
 \label{eq:gkp_moving_frame_derivative}
\end{equation}
where the dot on the right differentiates parameters at fixed $u$.
Write $G=G_{\Omega,s}$, $\Theta(u)=\vartheta_0(u-s,\Omega)$ and set
\begin{equation}
 B=(\Im\Omega)^{-1},\quad \tau=\Im s,\quad
 \nu=\Im\sigma,\quad \dot B=-B(\Im H)B.
\end{equation}
As in the preceding proof, for fixed $u$ we have
\begin{align}
 \ell(u):=\frac{\dot G(u)}{G(u)}
 =\frac{\pi}{2}u^T\dot Bu-2\pi i(u-s/2)^T(\dot B\tau+B\nu)
 +i\pi\sigma^TB\tau
\end{align}
and if instead we held $z_0=Xu$ fixed, we would obtain
\begin{equation}
  \frac{d}{dt}\log G_{p_t}(u_t) =\ell(u)-(A_Xu)^T\nabla_u\log G(u)\eqqcolon P(u),
\end{equation}
where $P(u)$ is a polynomial of degree at most two.
Let $\mathcal R(u)=P(u)\Theta(u)+\dot\Theta(u)
-(A_Xu)^T\nabla\Theta(u)$, so that at fixed $z=Xu$, we have $\dot f^B_{0,p}(Xu)= G(u)\CR(u)$. Substituting this into eq.~\eqref{eq:gkp_moving_frame_derivative} yields
\begin{equation}
 \frac{L(u,\overline v)}{G(u)\overline{G(v)}}
 =\mathcal R(u)\overline{\Theta(v)}
  +\Theta(u)\overline{\mathcal R(v)}.
\end{equation}
Since $L$ is in the normalized kernel span, this quotient is
$D\Z^n$-periodic in $u$. The functions $\Theta$, $\dot\Theta$ and
$\nabla\Theta$ have the same periods. Thus, for $t\in D\Z^n$,
\begin{equation}\label{eq:gkp_frame_period_difference}
 \bigl[\underbrace{(P(u+t)-P(u))}_{
    \Delta_tP(u)
 }\Theta(u)
       -(A_Xt)^T\nabla\Theta(u)\bigr]\overline{\Theta(v)}=0.
\end{equation}
Taking another difference by $t'\in D\Z^n$ gives
$(\Delta_{t'}\Delta_tP(u))\Theta(u)\overline{\Theta(v)}=0$.
The theta function is non-zero, and $D\Z^n$ spans $\R^n$, so the quadratic part of $P$ vanishes. Consequently $P(u+t)-P(u)=c^Tt$ with constant $c\in \C^n$, and eq.~\eqref{eq:gkp_frame_period_difference} implies
\begin{equation}
 (A_Xt)^T\nabla\Theta(u)=c^Tt\Theta(u).
\end{equation}
Since the coefficients of $\Theta(u)$ are nonzero, we find that $c^T =0$ and hence also $A_X=\dot{X}=0$. Comparing those at $r=e_j$ gives $A_Xt=0$.

Thus $L=\dot K_{p,00}$, and $P=\ell$ is a $D\Z^n$-periodic
polynomial, hence constant. Its quadratic and linear coefficients
give $\dot B=0$ and $B\nu=0$, respectively.
Thus $\Im H=\Im\sigma=0$.

To show that the assumption also implies that $\Re H=\Re\sigma=0$, note that $\ell=i\pi\sigma^TB\tau$ is purely
imaginary, so that $\dot{(G(u)\overline{G(v)})}=0$. 
The quotient $\dot K_{p,00}/(G(u)\overline{G(v)})$ has only integer Fourier frequencies, and therefore we have $\dot K_{p,00}=aK_{p,00}$ for some $a\in\C$. %
Since the constant Fourier coefficient of $\Theta(u)\overline{\Theta(v)}$ is $c_0(p)\overline{c_0(p)}=1$, we moreover have $a=0$. 
Consider the $(r,0)$ element of $K_{p,00}/(G(u)\overline{G(v)})$ with basis $e^{i\pi r^T\Omega r-2\pi i r^Ts} \ne0$. Since we have established that $\dot K_{p,00} =0$ and $\dot{(G(u)\overline{G(v)})}=0$, we must have
\begin{equation}
\dot{c}_r(p) = i\pi(r^THr-2r^T\sigma)c_r(p)=0,\qquad r\in\Z^n.
\end{equation}
Taking $r=\pm e_j$ gives $H_{jj}=\sigma_j=0$, and $r=e_i+e_j$
gives $H_{ij}=0$ for $i\ne j$. Thus $h=0$.

We are almost done. Finally, define the norm on operator-valued $1$-forms
\begin{equation}
    \|d\Phi(p)\|_{\rm op}\coloneqq\sup_{\substack{\|h\|=1\\ \|L\|_{H_B, \beta}=1}} \|d\Phi(p)[h]L\|_{H_B, \beta}.
\end{equation}
By finite dimensionality of the tangent space containing $h$, there is a lower bound 
\begin{equation}
 m:=\min_{\substack{\|h\|=1\\h=(H,\sigma,0)}}\|\mathrm d\Phi(p_0)[h]\|_{\mathrm{op}}>0.
\end{equation}
In the previous lemma~\ref{thm:gkp_kernel_lipschitz} we have shown that $p\to \Phi(p)$ is a smooth map. Consequently, $p\to d\Phi(p)$ is also continuous relative to the norm above.
We can thus shrink the convex neighbourhood $U$ so that
\begin{equation}
   \sup_{p\in U}\|\mathrm d\Phi(p)-\mathrm d\Phi(p_0)\|_{\rm op}\leq m/2.
\end{equation}

For $h=p'-p$, integration along the segment from $p$ to $p'$ and the reverse triangle inequality give
\begin{align}
 \|\Phi(p')-\Phi(p)\|_{\mathrm{op}}
 &=\left\| d\Phi(p_0)[h]+\int_0^1(d\Phi(p+th)-d\Phi(p_0))[h]\,dt\right\|_{\mathrm{op}}
 \nonumber\\
 &\geq\|\mathrm d\Phi(p_0)[h]\|_{\mathrm{op}}-\frac m2\|h\|
 \nonumber\\
 &\geq\frac m2\|p'-p\|,
\end{align}
which yields the claimed bound with $c=m/2$.

Finally, choose $X=R(\eta)X_0(\Omega,s)$, take a joint direction
$h=(H,\sigma,\xi)$ and assume $\mathrm d\Phi(p)[h]=0$.
The kernel calculation above uses the frame variation only through
$A_X=X^{-1}\dot X$, so it applies to this choice and gives
$\dot X=0$ and $H=\sigma=0$. Hence
\begin{equation}
 0=\dot X=\mathrm dR(\eta)[\xi]X_0(\Omega,s).
\end{equation}
Since $X_0(\Omega,s)$ is invertible and $R$ is a coordinate parametrization, $\xi=0$. The full differential is therefore injective. Smoothness from the preceding lemma and the same minimum and integration argument, now over all unit directions $h=(H,\sigma,\xi)$, give the stated bound after possibly shrinking $U\subset\mathcal U_0$ further.
\end{proof}

The previous proofs exploited the explicit analytic structure we have for the kernels $K_{p,a}$ and showed that there is a well-defined continuous map between the operator-algebra code spaces for the GKP codes and the parameters defining their stabilizer groups. In particular, this means that we can identify the topology of the space of GKP operator algebras with that of its defining parameters, which we now review.

\subsubsection*{GKP parameters revisited}
We define the group of polarization-preserving transformations as~\cite{Birkenhake_2004,conrad2024latticesgatescurvesgkp}
\begin{equation}
    \Sp_D(\R)=\{S\in \GL_{2n}(\R),\, SJ_DS^T=J_D\},\quad  J_D=(J_2\otimes D),
\end{equation}
with subgroup $\Gamma_D = \Sp_D(\Z)$ and write $\T^{2n}=\R^{2n}/\Z^{2n}$ for the real torus.
The parameter space for all GKP codes with fixed type $D$ is~\cite{Burchards_2025,conrad2024latticesgatescurvesgkp, Birkenhake_2004,mayrand2026complex}
\begin{equation}
    \CM_D=\Gamma_D\setminus \lr{\Sp_D(\R)\times \T^{2n}},\quad \gamma\cdot(\mathsf{S}, \theta)=(\gamma\mathsf S, \gamma\theta+\ell_{\gamma}/2),
\end{equation}
where $\ell_{\gamma}=\diag\lr{\gamma \lr{\begin{smallmatrix}0 & D \\ 0 & 0\end{smallmatrix}}\gamma^T}$ tracks the correct transformation of the reference semicharacter $\chi_{\theta}(k)=(-1)^{m^TDn}e^{i2\pi \theta^Tk}$ for $k=(m, n)\in \Z^{2n}$. The quotient $\CM_D$ captures the space of all symplectically integral lattices of type $D$, which carries a transitive action of $\Sp_D(\R)$, and contains the syndrome spaces via shifts in $\T^{2n}$.

To recover our holomorphic coordinates, let $\sigma_D=I_n\oplus D$, so that for $\mathrm{S}\in\Sp_D(\R) $
\begin{equation}
    S=\sigma_D^{-1}\mathrm{S}\sigma_D=\lr{\begin{smallmatrix}
        a & b \\ c & d
    \end{smallmatrix}}\in \Sp_{2n}(\R).
\end{equation}
and we have, in coordinates where $\Omega_B=iI_n$, the identification
\begin{equation}
    (iI, I )S^T\sigma_D = X(\Omega, D),\qquad \Omega = S.(iI_n),\;X=(ci+d)^T.
\end{equation}
To extract the rotation coordinate, we fix the reference matrices \begin{equation}Y_0=D^{-\frac{1}{2}}(\Im\Omega)D^{-\frac{1}{2}},\quad  X_0=Y_0^{-\frac{1}{2}}D^{-\frac{1}{2}},\end{equation}
so that $R=XX_0^{-1}\in U(n)$. Similarly, the syndrome is contained in the torus 
\begin{equation}s\in X^{\perp}_{\Omega, X}=\C^n/(\Omega D^{-1}\Z^n +\Z^n)\sim \T^{2n}.\end{equation}
Physical rotations $U\in U(n)$ act by $(\Omega, s, R)\mapsto (\Omega, s, UR)$. Forgetting the orientation label $R$, and then the syndrome parameter, identifies GKP codes with the orbifold quotients~\cite{hain2014lectures,conrad2024latticesgatescurvesgkp}
\begin{equation}
    \CJ_D = \CM_D/\!/U(n)= \Gamma_D\mathord{\backslash\!\backslash}\HH_n\times \T^{2n} \longrightarrow \CA_D = \Gamma_D\mathord{\backslash\!\backslash}\HH_n.
\end{equation}
Each point in $\CJ_D$ retains information about its stabilizer: writing $\gamma_D=\sigma_{D}^{-1}\gamma \sigma_D$,
\begin{equation}
    \Gamma_{p}=\lrc{\gamma\in \Gamma_D: \gamma_D.\Omega=\Omega,\; \gamma\theta + \ell_{\gamma}/2=\theta\; \mod \Z^{2n} }, \label{eq:orbifold points}
\end{equation}
and for every $\gamma\in\Gamma_{p}$, writing $\gamma_D=\lr{\begin{smallmatrix}a_{\gamma}&b_{\gamma}\\ c_{\gamma} & d_{\gamma}\end{smallmatrix}}$, one has
\begin{equation}
    U_{\gamma}=X(c_{\gamma}\Omega+d_{\gamma})^{-T}X^{-1} \in U(n),\label{eq:unitary_auto}
\end{equation}
which is a unitary automorphism of the GKP lattice $U_{\gamma}\Lambda=\Lambda$ and semicharacter $\chi_s(U_{\gamma}\lambda)=\chi_s(\lambda)$, i.e. preserves the associated stabilizer group.
Unitarity can be verified by setting $Y=\Im\Omega$, $A=(c_{\gamma}\Omega+d_{\gamma})^{-T}$ and using the fact that we have chosen a frame satisfying $XYX^{\dagger}=I$.

\begin{theorem}[Global smooth embedding]\label{thm:gkp_local_embedding}
Fix a polarization type $D$ with $d=\det D$, and let $\beta>0$.
For $p\in\CM_D$, let $V_p$ be the physical kernel space in
eq.~\eqref{eq:gkp_physical_kernel_span}. The map
\begin{equation}
 F:\CM_D\longrightarrow\operatorname{Gr}_{d^2}(\mathscr K_{H_B,\beta}),
 \qquad p\longmapsto V_p
\end{equation}
is a smooth embedding.
Every $p_0\in\CM_D$ has a coordinate neighbourhood $U$ on which $F$
is bi-Lipschitz.
\end{theorem}
\begin{proof}
The free and properly discontinuous $\Gamma_D$-action makes $\CM_D$ a smooth manifold, locally identified with the parameter charts above. Since rebasing preserves $V_p$, the map $F$ descends to $\mathcal M_D$. Lemmas~\ref{thm:gkp_kernel_lipschitz}
and~\ref{lem:gkp_kernel_lower} make $F$ a smooth immersion and give the
local bi-Lipschitz bounds. It remains to show that $F$ is globally injective and has a continuous inverse.

Write $\Lambda_p$ for the physical stabilizer lattice, with $T(\lambda)$
acting on the left kernel variable. We recover it from $V_p$ as
\begin{equation}\label{eq:gkp_scalar_stabilizer_recovery}
 \Lambda_p=\{\lambda:T(\lambda)|_{V_p}=cI,\quad c\in U(1)\}.
\end{equation}
Indeed, every $T(\mu)$ with $\mu\in\Lambda_p^\perp$ preserves $V_p$.
A scalar $T(\lambda)|_{V_p}$ commutes with all these restrictions, so that $\Im H_B(\lambda,\Lambda_p^\perp)\subset\Z$,
hence $\lambda\in(\Lambda_p^\perp)^\perp=\Lambda_p$.
The converse is the stabilizer equation. Its scalar eigenvalues also recover the semicharacter, so $F$ is injective.

To show the continuity of the inverse, consider a sequence $V_{p_j}\to V_p$, and write $\Phi_j,\Phi$ for their orthogonal projections. 
For every $K\in V_p$, we have $\Phi_jK\to K$ in $\mathscr K_{H_B,\beta}$ and locally uniformly by the estimate in Lemma~\ref{lem:kernel_space_complete}. We must show that $V_{p_j}\to V_p$ inside $\mathscr K_{H_B,\beta}$ rules out a sequence of nonzero lattice vectors $\lambda_j\in\Lambda_{p_j}$ such that $\lambda_j\to 0$.

If nonzero $\lambda_j\in\Lambda_{p_j}$ tended to zero, consider their normalized limiting direction $v$ with $\lambda_j/\|\lambda_j\|\to v$, $\|v\|=1$. Let $r>0$ be smaller than every nonzero vector length in $\Lambda_p$ and let $m_j=\lfloor r/\|\lambda_j\|\rfloor$ such that $m_j\|\lambda_j\|\to r$.

Then $\mu_j=m_j\lambda_j\to rv$ and $T(\mu_j)|_{V_{p_j}}=c_jI$. 
Taking the limit $c_j\to c\in U(1)$ uniformly on compact sets in $(z, \overline{w})$, one then obtains
\begin{equation}
 T(\mu_j)\Phi_jK=c_j\Phi_jK
 \quad\longrightarrow\quad T(rv)K=cK,\qquad K\in V_p.
\end{equation}
Equation~\eqref{eq:gkp_scalar_stabilizer_recovery} now implies $rv\in\Lambda_p$, contradicting the fact that $r$ is smaller than every non-zero vector length.
Thus the shortest nonzero lattice vectors are bounded away from zero. Since all covolumes equal $\det D$, Mahler's compactness theorem \cite[Theorem~2]{Mahler1946} gives convergent lattice bases along a subsequence. Their integral symplectic Gram matrices are eventually
constant, so a single integral change of basis makes them convergent symplectic bases of type $D$. The syndromes lie in the compact torus $\T^{2n}$, hence the parameter classes $p_j$ are relatively compact in $\CM_D$. 
Every cluster point $q$ satisfies $F(q)=F(p)$ by continuity, and therefore $q=p$ by injectivity.  We have thus shown that $F(p_j)\to F(p)\Rightarrow p_j\to p$, which proves that $F$ is a homeomorphism onto its image and hence a smooth embedding.
\end{proof}

The preceding statements have shown that the quotient topology on $\CM_D$ agrees with the topology induced by the regularized distance $d_{H_B, \beta}$ on $\Gr_{d^2}(\mathscr{K}_{H_B, \beta})$. 

\subsection{The Theta bundle and Chern connection}
In this subsection we fix a rotation $R_0$ and use $X(\Omega,s)=R_0X_0(\Omega,s)$, suppressing the fixed rotation in the kernel and subspace labels. We show that the theta bundle induces a flat Chern connection on this fixed-orientation family, embedded by
Theorem~\ref{thm:gkp_local_embedding}.

Locally over $\mathcal J_D$, the spaces introduced in the previous subsection form the theta bundle
\begin{equation}
    \mathcal E_{D,(\Omega,s)}
    =
    H^0(X_{\Omega,D},L_s),
    \qquad
    \operatorname{rank}\mathcal E_D=\det D,
\end{equation}
and the corresponding logical operator bundle is
$\operatorname{End}(\mathcal E_D)$.
Using the fact that theta functions satisfy the heat equation~\cite{Mumford1966}, Hitchin~\cite{Hitchin1990FlatConnections} and later To and Weng~\cite{ToWeng2004} noted that the vector bundle spanned by theta functions admits a projectively flat connection. This observation immediately translates into flatness for the algebra of encoded logical operators for GKP codes parametrized by $\CJ_D$.
\begin{proposition}\label{prop:theta_end_flat}
Over the complex orbifold $\CJ_D$, the intrinsic Hermitian structure on the theta
bundle induces a flat Chern connection on $\operatorname{End}(\mathcal E_D)$.
\end{proposition}

\begin{proof}
On a local chart, let $Y=\Im\Omega$ and choose the holomorphic
theta frame
\begin{equation}
    f_\alpha(z;\Omega,s)=\vartheta_\alpha(z-s,\Omega),
    \qquad
    \alpha\in D^{-1}\Z^n/\Z^n,
\end{equation}
whose Hermitian matrix is~\cite[prop. 8.5.6]{Birkenhake_2004}
\begin{equation}
    \mathsf H_{\alpha\beta}:=\langle f_\alpha,f_\beta\rangle_{\Omega,s}
    =
    h(\Omega,s)\delta_{\alpha\beta},
    \qquad
    h(\Omega,s)
    =
    \frac{\det D}{2^{n/2}}\sqrt{\det Y}.
\end{equation}
Thus $\mathsf H=hI_d$, where $d=\det D$. In a holomorphic frame, the
Chern connection has matrix~\cite[Ch.~V, \S12, eq.~(12.2)]{Demailly2012CADG}
\begin{equation}
    \mathsf A=\mathsf H^{-1}\partial\mathsf H
    =(\partial\log h)I_d.
\end{equation}
Writing an endomorphism in this frame as a coefficient matrix $M$, we obtain the
induced Chern connection
\begin{equation}
    \nabla^{\operatorname{End}(\mathcal E_D)}M
    =\mathrm dM+\mathsf A M-M\mathsf A=\mathrm dM,
\end{equation}
and is therefore flat. Theta identifications act on $\mathcal E_D$ up to
scalar and preserve its Hermitian metric up to scale. They induce holomorphic isometries on $\operatorname{End}(\mathcal E_D)$, so uniqueness of the Chern connection makes these local connections agree on orbifold charts.
\end{proof}

To connect this observation to the operator frames developed earlier, let $p=(\Omega, s)\in \CJ_D$ be the code label and let \begin{equation}
    \psi_{p, \alpha}(z)=\sqrt{\det(\Im\Omega) /h(\Omega, s)}G_p(z)f_{p, \alpha}(z).
\end{equation}
Note that $\psi_{p, \alpha}(z)$ is non-holomorphic in $\Omega$. The kernels 
\begin{equation}
    E_{p,\alpha\beta}(z, \overline{w})= \psi_{p, \alpha}(z)\overline{\psi_{p, \beta}(w)}
\end{equation}
form an orthonormal operator frame basis with $\langle E_{p,\alpha\beta}, E_{p,\gamma\delta}\rangle_{{\rm HS}, p}=\delta_{\alpha\gamma}\delta_{\beta\delta}$, and we can write
\begin{equation}
    \CK_{\Delta_p(a)} = \sum_{\alpha, \beta} a_{\alpha\beta}E_{p,\alpha\beta},\quad  \Lambda_p(\cdot) = \sum_{\alpha\beta}e_{\alpha\beta}\langle E_{p,\alpha\beta}, \cdot \rangle_{{\rm HS}, p}.
\end{equation}
In physical coordinates this operator frame is
\begin{equation}\label{eq:gkp_physical_operator_frame}
 E^B_{p,\alpha\beta}(z_B,\overline{w_B})
 =E_{p,\alpha\beta}\bigl(X(p)^{-1}z_B,
                  \overline{X(p)^{-1}w_B}\bigr)
\end{equation}
and these kernels span $V_p$.
The induced principal connection is therefore locally trivial
\begin{equation}
    \omega_{\rm Ch}
    =
    \Lambda_p
    \nabla^{\operatorname{End}(\mathcal E_D)}
    \Delta_p
    =
    0.
\end{equation}

Indeed, after pullback to $\CM_D$ and identification with the physical
logical operator bundle, this Chern connection agrees with the flat
connection induced by Gaussian transport in ref.~\cite{Burchards_2025}. Together with Theorem~\ref{thm:gkp_local_embedding}, this identifies the Chern transport with a flat transport of logical operator frames over the GKP submanifold of $\operatorname{Gr}_{d^2}(\mathscr K_{H_B,\beta})$.

\section{Geometry of magic states}\label{sec:magic}
The previous section has shown that, for a fixed $\Omega \in \CA_D$, the GKP code space for syndrome $s\in X_{\Omega, D}$ in Segal-Bargmann representation is spanned by a basis for $H^0(X_{\Omega, D}, L_s)$ given by the functions
\begin{equation}
    \psi_{(s, \Omega), \alpha}(z)=G_{\Omega, s}(z)\vartheta_{\alpha}(z-s, \Omega),\quad \alpha\in D^{-1}\Z^n/\Z^n,
\end{equation}
where we consider normalized coordinates. With $F(z)=1$, the Segal-Bargmann wavefunction for a coherent state at ${\tt c}\in \C^n$ is given by $k_{{\tt c}}(z)=e^{-\frac{\pi}{2} H_{\Omega}({\tt c},{\tt c})-\pi H({\tt c}, z)}$ and the projection of a coherent state $\ket{{\tt c}}$ onto the GKP code space for the lattice indexed by $\Omega$ yields the logical state
\begin{equation}
    \ket{\psi_{\Omega}({\tt c})}=\frac{\sum_{\alpha\in D^{-1}\Z^n/\Z^n}\overline{\vartheta_{\alpha}(-{\tt c}, \Omega)}\ket{\alpha}}{\sqrt{\sum_{\alpha\in D^{-1}\Z^n/\Z^n}|\vartheta_{\alpha}(-{\tt c}, \Omega)|^2}},\label{eq:GKP_proj_coh}
\end{equation}
assuming the projection is nonzero. The above equation can also be understood as obtaining $\ket{\psi_{\Omega}({\tt c})}$ by projecting the vacuum state $\ket{\tt 0}$ onto the syndrome sector $s={\tt c}$. In particular, for ${\tt c}=\tt 0$, the corresponding logical state amplitudes are given by $\overline{\vartheta_{\alpha}(-{\tt c}, \Omega)}$. 

A special feature of the vacuum $\ket{{\tt 0}}$ is that it is a rotation-invariant quantum state, and the corresponding Segal-Bargmann function $F(z)=1$ is constant. In eqs.~\eqref{eq:orbifold points},~\eqref{eq:unitary_auto}, we pointed out that orbifold points of $\CJ_D$ are equipped with unitary automorphisms $U_{\gamma}\in U(n),\; \gamma\in \Gamma_p$. As the vacuum is unitarily invariant, this observation implies that orbifold points produce states $ \ket{\psi_{\Omega}({\tt c})}$ with additional symmetries.

Quantum states whose stabilizers necessarily include non-Pauli operators are called magic states and are a central resource in common approaches to fault-tolerant quantum computation, including those using the GKP code~\cite{BravyiKitaev, GKP}. In the following, we show that GKP codes with symplectic orthogonal automorphisms can arise from algebraic curves via their Jacobians and automorphisms, which extends the observation for the square GKP code from ref.~\cite{baragiola2019all}.

We first choose a new set of parameters. A real GKP lattice $S^TL_D$, with $L_D=\Z^n\oplus D\Z^n$, has normalized period $\Omega=S.(iI_n)$ and physical frame $X=(ic+d)^T$. Relative to the reference lattice $L_{\sqrt{D}}=\sqrt{D}\Z^n\oplus \sqrt{D}\Z^n$ we instead consider the principal period and frame $\tau\in \HH_n, X_{\tau}\in \GL_n(\C)$, which satisfy 
\begin{equation}
    \Omega = \sqrt{D}\tau \sqrt{D}=(\sqrt{D}\oplus \sqrt{D}^{-1}).\tau,\quad X=X_{\tau}D^{-\frac{1}{2}}.
\end{equation}
In particular, for $D=dI_n$, $d\in \Z_{\geq 1}$, we set $\Omega = d\tau$. 

\subsection{Jacobians}
In this introduction, we closely follow the discussion in ref.~\cite{Birkenhake_2004}. Conceptually, Jacobians are complex tori obtained via an ``abelianization'' of a compact Riemann surface, or curve $C$. Let $H^0(\omega_C)$ denote the space of holomorphic $1-$forms on $C$. We identify the genus of the curve $g=n=\dim H^0(\omega_C)$ with the number of modes. Paths $\gamma$ on the curve are identified with elements $p_{\gamma}\in H^0(\omega_C)^*$ via
\begin{equation}
    p_{\gamma}:H^0(\omega_C)\to \C, \quad \omega \mapsto \int_{\gamma}\omega.
\end{equation}
Let $\lrc{\omega_i},i=1,\hdots, g=\dim H^0(\omega_C)$ be a basis for the space of holomorphic $1$-forms of a genus $g$ curve $C$, which we choose to be orthonormal with respect to the Hermitian form
\begin{equation}
    (\omega_i, \omega_j) = \frac{-i}{2}\int_C \overline{\omega_i} \wedge \omega_j =\delta_{ij}.
\end{equation}
With $\vec{\omega}=(\omega_1,\hdots, \omega_g)^T$, every path $\gamma: [0,1]\to C$ 
is mapped to a vector 
\begin{equation}
    {\rm AJ}: H^0(\omega_C)^* \to \C^{g},\quad \gamma \mapsto \int_{\gamma}\vec{\omega}=\lr{\int_{\gamma}\omega_i}_{i=1\hdots g}.
\end{equation}
This map is well-defined up to the Jacobian lattice
\begin{equation}
    \Lambda_J = \lrc{\int_{\gamma}\vec{\omega}\in \C^g,\; \gamma \in H_1(C, \Z)}, 
\end{equation}
which is generated by a basis $\Pi\in \C^{g\times 2g}$ with $\Pi_{ij}=\int_{\gamma_j}\omega_i$ as $\gamma_j$ runs over a generating set for the $1^{\rm st}$ homology group $H_1(C,\Z)$.
Let $P(\gamma_{i})$ denote the Poincar\'e dual of $\gamma_i$, so that $\int_{\gamma_i}\omega = i2(P(\gamma_i), \omega)\; \forall \omega \in H^0(\omega_C)$. The Hermitian form $\Pi^{\dagger}\Pi$ can be understood as the inner products of the holomorphic projections of $P(\gamma_i)$, resulting in


\begin{equation}
    \lr{\Pi^{\dagger}\Pi}_{ij} = \int_C P(\gamma_i)\wedge * P(\gamma_j) + i \int_C P(\gamma_i)\wedge  P(\gamma_j). \label{eq:Hermitian_Jacobian}
\end{equation}
In this equation, $\int_C P(\gamma_i)\wedge  P(\gamma_j) = (\gamma_i \cdot \gamma_j)$ is the intersection pairing and $*$ is the Hodge star.
For a canonical homology basis $\lrc{\gamma_i}_{i=1...2g}$, we have $ (\gamma_i \cdot \gamma_j)=J_{ij}$, so that  the resulting lattice is symplectic and defines a (principally) polarized abelian variety. After interchanging the two blocks of the canonical homology basis, so that the intersection matrix becomes $-J$, one can normalize the basis $\lrc{\omega_i}_{i=1\hdots g}$ such that $\Pi=(\Omega, I_g)$, which yields the Hermitian form $H_{\Omega}$ used above.  

Automorphisms of Riemann surfaces $f\in Aut(C)$ are represented by unitary automorphisms of the Jacobian lattice $\Lambda_J$.
Because $f:C\to C$ is holomorphic, for all $a, b\in H^0(\omega_C)$, $f$ preserves the inner product, $(f^*a, f^*b)=(a,b)$, and hence also acts unitarily on $\C^g$. 


Before discussing the specific examples, we summarize the results of this section informally:
\begin{proposition}
GKP magic states can be obtained by projecting vacuum states onto the GKP code space associated to a GKP Jacobian code. In particular, we find 
\begin{enumerate}
\item single-qubit magic states 
\begin{equation}
    \ket{\psi_{2\tau}(\tt 0)} =\frac{1}{\sqrt{1+|\lambda\lr{2\tau}|^{\frac{1}{2}}}} \lr{\ket{0}+\overline{\lambda\lr{2\tau}^{\frac{1}{4}}}\ket{1}}
\end{equation}

for elliptic curves $E_{\tau}:\, y^2=x(x-1)(x-\lambda(\tau))$ at $\tt c =0$, where the fourth root is chosen so that $\lambda(z)^{1/4}=\vartheta_{1/2,0}(0,z)/\vartheta_{0,0}(0,z)$.

\item the $n=3$ Klein quartic state
\begin{equation}
    \ket{\psi_K}=\frac{1}{\sqrt{672}}\lr{16\ket{000}+\sqrt{3}(11+i\sqrt{7})\ket{W_3} -(5-i\sqrt{7})\ket{111}}
\end{equation}
with $\ket{W_3}=\frac{1}{\sqrt{3}}(\ket{001}+\ket{010}+\ket{100})$ for Klein's quartic $C_K:\, x^3y+y^3z+z^3x=0$ at $\tt c =0$, and \item the $n=3$ CCZ state
\begin{equation}
    \ket{\tt CCZ} = \frac{1}{\sqrt{8}}\sum_{\sigma\in \Z_2^3} (-1)^{\sigma_1\sigma_2\sigma_3}\ket{\sigma}
\end{equation}
for the octahedral curve $C_O:\ y^2=x^8+14x^4+1$ at ${\tt c} =2(\frac{it}{4}+\frac{1}{8})\bs{1}_3.$

In particular, $(1.)$ also implies that the projection of the $n-$mode vacuum onto the code space of a self-dual CSS code with $k=1$, concatenated with the square  ($\tau=i$) GKP code, yields a logical magic state.
\end{enumerate}

\end{proposition}

\subsection{Examples}

\subsubsection{The square and hexagonal magic states}
The two simplest examples occur when $n=1$ for $D=(2)$ and $\Omega=2 \tau$ with $\tau = i$ and $\tau=\rho \coloneqq e^{i2\pi/3}$, which have symmetries given by
\begin{equation}
    S=\begin{pmatrix}
        0 & -1 \\ 1 & 0
    \end{pmatrix},\quad T=\begin{pmatrix}
        1 & 1 \\ 0& 1
    \end{pmatrix}
\end{equation}
since
\begin{equation}
    S.i = i,\quad (ST).\rho =\rho.\label{eq:square_hex_sym}
\end{equation}
Evaluating eq.~\eqref{eq:GKP_proj_coh} for ${\tt c}=0$, we find
\begin{align}
    \ket{\psi_{2i}(\tt 0)} = N\lr{\ket{0}+\overline{\lambda\lr{2i}^{\frac{1}{4}}}\ket{1}}=N\lr{\ket{0}+\tan\lr{\frac{\pi}{8}}\ket{1}}\\
    \ket{\psi_{2\rho}(\tt 0)} = N'\lr{\ket{0}+\overline{\lambda\lr{2\rho}^{\frac{1}{4}}}\ket{1}}=N'\lr{\ket{0}+\sqrt{2-\sqrt{3}}e^{i\frac{\pi}{4}}\ket{1}}
\end{align}
where we have used the modular $\lambda$ function given by $\lambda(z)=\vartheta_{\frac{1}{2},0}(0,z)^4/\vartheta_{0,0}(0,z)^4$ and $N, N'$ are normalization parameters. 

\subsubsection{The Klein quartic magic state}
The Klein quartic is the curve
\begin{equation}
K = \{[x:y:z] \in \mathbb{P}^2 : x^3 y + y^3 z + z^3 x = 0\}, \quad g=3.
\end{equation}
Its automorphism group is $\mathrm{Aut}(K) = \mathrm{PSL}_2(\F_7)$, of order $|\mathrm{PSL}_2(\F_7)|=168$, saturating the Hurwitz bound $|\mathrm{Aut}| \leq 84(g-1)$. According to ref.~\cite{Braden_2010}, the period matrix is
\begin{equation}
\tau_{K} = \frac{1}{2} \begin{pmatrix} \varepsilon & 1 & 1 \\ 1 & \varepsilon & 1 \\ 1 & 1 & \varepsilon \end{pmatrix}, \qquad \varepsilon = \frac{-1 + i\sqrt{7}}{2}.
\end{equation}

We now give an expression for the associated vacuum projection onto the $d=2$ GKP code. As $\tau_{K}$ is cyclic, the coefficients are also cyclic, so we only need to determine $c_{000},\, c_{100}$, $c_{110}$ and $c_{111}$. Computing the coefficients using theta identities (see appendix~\ref{app:Klein}) and normalizing yields
\begin{align}
    \ket{\psi_K}&=N''\sum_{\sigma\in \Z_2^3}\;\overline{\vartheta_{\sigma/2,0 }(0, 2\tau_{K})}\ket{\sigma}\nonumber\\
    &=\frac{1}{\sqrt{672}}\lr{16\ket{000}+\sqrt{3}(11+i\sqrt{7})\ket{W_3} -(5-i\sqrt{7})\ket{111}},\label{eq:310}
\end{align}
where $\ket{W_3}=\frac{1}{\sqrt{3}}(\ket{001}+\ket{010}+\ket{100})$. We recognize that $\ket{\psi_K}$ is a superposition of a GHZ-like state with the non-stabilizer Dicke state $\ket{W_3}$.

In a fixed standard Bargmann space with $Q=I_3$, this state is obtained from the code with stabilizer lattice $\Lambda=X(2\tau_{K}\vert I_3)(\Z^3\oplus2\Z^3)$ with $X=\sqrt{2/\sqrt{7}}I_3$. In these coordinates, the theta coordinate maps to $z_{B}=Xz$.

\subsubsection{The CCZ magic state}
Let $\Theta_{\sigma}(z_0, \tau)=\vartheta_{\sigma/2,0}(2z_0, 2\tau),\, \sigma\in \Z_2^3$.
For $\tau(t)=it I_3+\frac{1}{2}(\bs{1}_3\bs{1}^T_3 -I_3)$, $t>0$ and $z_0=(\frac{it}{4}+\frac{1}{8}) \bs{1}_3$, a calculation with $q=e^{-\pi t}$ yields
\begin{align}
    \Theta_{\sigma}(z_0, \tau(t)) &=e^{i\frac{\pi}{4}{{\rm wt}}(\sigma)}i^{\binom{\wt(\sigma)}{2}}\prod_{j=1}^3f_{\sigma_{j}}(\sigma, q), \nonumber\\ f_{\sigma_{j}}(\sigma, q)&=\sum_{n\in \Z} i^n(-1)^{n({\rm wt}(\sigma)-\sigma_j)}q^{2n^2+n(2\sigma_j+1)+\sigma_j}.
\end{align}
With
\begin{equation}
    u=\sum_{n\in \Z} i^nq^{2n^2+n},\quad v= \sum_{n\in\Z} (-i)^nq^{2n^2+n } 
\end{equation}
and $U=u^3, V=e^{-i\frac{\pi}{4}}v^3$, the coefficients are 
\begin{equation}
    (\Theta_{000}, \Theta_{001}, \Theta_{010},\hdots,\Theta_{111} ) = (U, V, V, U, V, U, U, -V).
\end{equation}
Now, using the theta duplication formulas~\cite{DLMF}, one finds that $U=V$ is achieved when $\lambda(i2t)=\frac{1}{4}$, such that
\begin{equation}
    \lr{\frac{v}{u}}^2=\sqrt{1-\lambda(i2t)}+i\sqrt{\lambda(i2t)}=e^{i\frac{\pi}{6}}.
\end{equation}
In ref.~\cite{ejiri2018morseindextriplyperiodic}, Ejiri and Shoda show that the Jacobian of the octahedral hyperelliptic curve
\begin{equation}
    C_O:\quad y^2=x^8+14x^4+1
\end{equation}
realizes these parameters. There, the authors choose a basis of holomorphic $1$-forms 
\begin{equation}
    \omega = (1-z^2\; i(1+z^2)\; 2z)^T \frac{dz}{\sqrt{z^8+14z^4+1}}
\end{equation}
to obtain a period matrix $\Pi_O$, for which a basis transformation via $O\in \GL_{6}\lr{\Z}$ yields
\begin{equation}
    \Pi_O O ^T=
\begin{pmatrix}
-iB & -A & iB & -iB & -i2B & -iB \\
A & iB & -A & iB & 0 & -iB \\ 
-iB & iB & -iB & A & 0 & A    
\end{pmatrix}O^T =
-i2B \lr{\tau\lr{t}\vert I_3}
\end{equation}
with $t=\frac{A}{2B}$ and
\begin{equation}
    A=K\lr{\sqrt{\frac{3}{4}}},\quad B = K\lr{\sqrt{\frac{1}{4}}},
\end{equation}
where 
\begin{equation}
    K(\sqrt{m}) = \int_{0}^{\infty} \frac{dt}{\sqrt{t^4+(2-4m)t^2+1}}
\end{equation}
is the Legendre integral for $0\leq m<1$ (see~\cite[19.2]{DLMF}).
Using the definitions of the modular lambda function and theta functions~\cite[22.2]{DLMF}, we obtain 
\begin{equation}
    \lambda\lr{i\frac{K(\sqrt{1-k^2})}{K(k)}} = k^2,
\end{equation}
and $C_O$ can be verified to yield the desired parameter $\lambda(i2t)=\frac{1}{4}$. 

Back in the Segal-Bargmann space with $Q=I$, we use $X=(2t)^{-1/2}I_3$ and $\Pi=X(2\tau(t)\mid I_3)$, such that the theta argument $2z_0$ corresponds to the physical evaluation point $z_{B}=\sqrt{2/t}z_0$.

\subsection{Concatenated GKP-CSS code}

A concatenated GKP-CSS code~\cite{Conrad_2021, NielsenChuang2010} is a GKP code given by the lattice
\begin{align}
    \Lambda_{d, \CT, H_X, H_Z} = \frac{1}{\sqrt{d}}\lr{\CT A_X\, \vert \, A_Z }\Z^{r_{X}+r_Z +2n},\nonumber\\
    A_X=(H_X^T | dI_n),\quad A_Z=(H_Z^T | dI_n)
\end{align}
for $\CT=\bigoplus_{i=1}^n \tau_i\in \HH_1^n$ in coordinates where the Hermitian form is $H_{\CT}$ and $X=(\Im\CT)^{-1/2}$.
The matrices $H_X\in \Z_d^{r_X \times n}$, $H_Z\in \Z_d^{r_Z \times n}$ define the parity-check matrix $H=H_X\oplus H_Z$, which satisfies
\begin{equation}
    HJH^T = 0 \mod d \Leftrightarrow H_XH_Z^T =0\mod d.
\end{equation}
The basis presented above is generically overcomplete and can be reduced using the Smith normal forms of $A_X, A_Z$. 

In the special case of a self-dual CSS code and $\CT=\tau^{\oplus n}$, we have $H_X=H_Z=H\in \Z_d^{r\times n}$. 
$A$ has Smith normal form $A=Q(D | 0_{n\times r})U$, where $Q\in \GL_{n}(\Z), U \in \GL_{n+r}(\Z)$ and $D$ is diagonal with integer entries $d_i|d_{i+1} | d,\, i=1,..,n$. In this case, the lattice can be written as
\begin{equation}
    \Lambda_{d, \tau, H} = \frac{1}{\sqrt{d}} QD \lr{ \tau\Z^n +\Z^n }.
\end{equation}
The structure of the CSS code is fully absorbed into the basis for $\C^n$, and, up to this basis transformation, the lattice is simply the $n$-fold copy of the single-mode lattice $\tau\Z+\Z$. A consequence of this simple structure is that for $\tau= i, \rho$ the symmetries $S.i=i$ and $(ST).\rho=\rho$ of eq.~\eqref{eq:square_hex_sym}, applied  simultaneously to each individual mode with parameter $\tau$, also realize symmetries of $\Lambda_{d, \tau, H}$.\\
Following the same reasoning as at the beginning of this section, this implies that for GKP-CSS codes whose CSS code is self-dual, the projection of the vacuum state onto the GKP code space for $\Lambda_{d, \tau, H}$ prepares an eigenstate of the logical Clifford gate associated to the symmetry when it also preserves the semicharacter as in eq.~\eqref{eq:unitary_auto}.

In simpler quantum error correction terms, take a self-dual CSS code encoding $k=1$ logical qubit into $n$ physical qubits, such as the $[\![7,1,3]\!]$ Steane code. For the square GKP code $\tau=i$, we can combine the fact that the square GKP code projects the vacuum onto a $\ket{H+}$ magic state with the existence of transversal Hadamard gates for self-dual CSS codes to conclude that projecting the $n$-mode vacuum onto the code space of the concatenated code results in a logical $\ket{\overline{H+}}$ state for the concatenated code.

\bibliographystyle{plain}
\bibliography{references}

\appendix
\phantomsection
\addcontentsline{toc}{section}{Appendix}
\addtocontents{toc}{\protect\setcounter{tocdepth}{0}}

\section{The Klein curve: coefficient computations}\label{app:Klein}

We evaluate the coefficients for the Klein magic state. Let $L_{a}=\lrc{k\in \Z^3:\; k \mod 2 = a}$. We need to evaluate
\begin{align}
    \vartheta_{a/2 ,0}(0,2\tau_K)
    &=\sum_{m\in \Z^3}e^{i\frac{\pi}{2}(2m+a)^T\tau_K(2m+a) } \nonumber\\
    &=\sum_{k\in L_{a}}e^{i\frac{\pi}{4}k^T\lrq{(\epsilon -1)I_3 + \bs{1}\bs{1}^T}k }\nonumber\\
    &=\sum_{k\in L_{a}}e^{i\frac{\pi}{4} (\epsilon -1)\|k\|^2} e^{i\frac{\pi}{4} (\bs{1}^Tk)^2}.\label{eq:125}\\
    &=\sum_{k\in L_{a}}e^{i\frac{\pi}{4} (\epsilon -1)\|k\|^2} e^{i\frac{\pi}{4} (\|k\|^2+ 2k_1k_2 +2k_1k_3+2k_2k_3)}\nonumber\\
    &=\sum_{k\in L_{\bs{a}}}e^{i\frac{\pi}{4} \epsilon \|k\|^2} e^{i\frac{\pi}{2} ( k_1k_2 +k_1k_3+k_2k_3)}. \label{eq:theta_form}
\end{align}
Now, for $a=(1\, 1\, 0)$,  $\bs{1}^Tk=2t$ is even. Using $(-1)^{t^2}=(-1)^t$ we can write the above sum (eq.~\eqref{eq:125}) as 
\begin{equation}
    \sum_{k\in L_{a}}e^{i\frac{\pi}{4} (\epsilon -1)k^Tk} (-1)^{\frac{\bs{1}^Tk}{2}}.
\end{equation}
$L_{a}$ is invariant under the orthogonal involution $O=\diag(-1,1,1)$, but $(-1)^{\frac{\bs{1}^TOk}{2}}=-(-1)^{\frac{\bs{1}^Tk}{2}}$, so we obtain $c_{110}=-c_{110}=0$ and consequently also $c_{011}=c_{101}=0$.
For the remaining coefficients, we have
\begin{align}
    \overline{c_{000}}
    &=\sum_{m\in \Z^3} e^{i\pi (\epsilon -1)\|m\|^2+i\pi (\bs{1}^Tm)^2}  \nonumber\\
    &=\sum_{m\in \Z^3} e^{i\pi \epsilon\|m\|^2 } = \vartheta(0,\epsilon)^3
\end{align}
where we used $(\bs{1}^Tm)^2=\|m\|^2 +2(m_1m_2+m_1m_3+m_2m_3)$.
Next, we have
\begin{align}
    \overline{c_{100}}&=\sum_{k\in L_{100}} e^{i\frac{\pi}{4}\epsilon\|k\|^2 +i\frac{\pi}{2}(k_1k_2+k_1k_3+k_2k_3)} \nonumber\\
    &=\vartheta_{\frac{1}{2},0}(0,\epsilon)\vartheta_{0,\frac{1}{2}}(0,\epsilon)^2.
\end{align}
and, finally,
\begin{align}
    \overline{c_{111}}&=\sum_{m\in \Z^3} e^{i\frac{\pi}{4}(\epsilon-1) \|2m+\bs{1}\|^2 +i\frac{\pi}{4}(2 \bs{1}^Tm+3)^2} \nonumber\\
    &=\sum_{m\in \Z^3} e^{i\pi \epsilon\| m+\frac{\bs{1}}{2} \|^2} e^{i\pi \lr{- \| m+\frac{\bs{1}}{2} \|^2 +(\bs{1}^Tm)^2 +3(\bs{1}^T\bs{m}) + \frac{9}{4} }} \nonumber\\
    &=-i\vartheta_{\frac{1}{2},0}(0,\epsilon)^3.
\end{align}

To evaluate the values stated in eq.~\eqref{eq:310}, we identify and set \begin{equation}A=\theta_3(\epsilon)=\vartheta(0,\epsilon),\quad B=\theta_2(\epsilon)=\vartheta_{\frac{1}{2},0}(0,\epsilon),\quad D=\theta_4(\epsilon)=\vartheta_{0,\frac{1}{2}}(0,\epsilon). \end{equation}
Since it holds that
\begin{equation}
    \epsilon^2+\epsilon+2=0\Rightarrow \frac{\epsilon}{2}=-\frac{1}{\epsilon+1}, 
\end{equation}
We may use the known relations between theta functions to solve for our coefficients. The relevant identities are the duplication formulas~\cite[20.7.12, 20.7.14]{DLMF} 
\begin{equation}
    \theta_2(\tau/2)^2=2\theta_2(\tau)\theta_3(\tau),\quad \theta_3(\tau/2)^2 = \theta_3(\tau)^2+\theta_2(\tau)^2,
\end{equation}
and modular properties of the theta functions
\begin{equation}
    \theta_2(-\frac{1}{\tau})^2=-i\tau\theta_4(\tau)^2,\quad \theta_3(-1/\tau)^2=-i\tau\theta_3(\tau)^2.
\end{equation}
Together with $\theta_3(\tau+1)=\theta_4(\tau)$ and $\theta_4(\tau+1)=\theta_3(\tau)$ these imply that
\begin{equation}
    2BA=-i(\epsilon+1)A^2,\quad A^2+B^2=-i(\epsilon+1)D^2.
\end{equation}
Calling $u=B/A$ and $v=D^2/A^2$ these equations can be solved to show that
\begin{equation}
    u=-i\frac{1+i\sqrt{7}}{4},\quad uv= \frac{11-i\sqrt{7}}{16},
\end{equation}
which implies the stated coefficients in eq.~\eqref{eq:310} after insertion and normalization.

\end{document}